\documentclass[letterpaper,12pt]{article}  

\usepackage{abstract}
\usepackage{amsfonts}
\usepackage{amsmath,amssymb,amsthm,mathabx}

\usepackage[title]{appendix}

\usepackage{tabularx,booktabs}
\newcolumntype{Y}{>{\centering\arraybackslash}X}

\usepackage{caption}
\usepackage{enumitem}

\usepackage{environ}
\NewEnviron{smallequation}{%
\begin{equation}
\scalebox{0.91}{$\BODY$}
\end{equation}
}

\usepackage{geometry}
\usepackage[pdftex]{graphicx}
\usepackage{hhline}  

\usepackage[colorlinks=true,citecolor=blue,linkcolor=blue,urlcolor=blue]{hyperref}

\usepackage{multirow}

\usepackage{natbib}
\bibpunct{(}{)}{;}{a}{}{,}

\usepackage{pdfsync}        

\usepackage{pgfplots}
\pgfplotsset{compat=1.12}
\usepgfplotslibrary{fillbetween}

\usepackage{sectsty}
\allsectionsfont{\singlespacing \raggedright}

\usepackage{setspace} 
\usepackage{subfigure}                     
\usepackage{tabularx}

\usepackage[flushleft]{threeparttable}
\usepackage{xr}
\usepackage{adjustbox}

\newtheoremstyle{myplain}
  {9pt}
  {9pt}
  {\itshape}
  {\parindent}
  {\scshape}
  {:}
  {.5em}
  {}

\newtheoremstyle{myremark}
  {9pt}
  {9pt}
  {}
  {\parindent}
  {\scshape}
  {:}
  {.5em}
  {}

\theoremstyle{myplain}
\newtheorem{assumption}{Assumption}[section]

\newtheorem{corollary}{Corollary}[section]
\newtheorem{definition}{Definition}[section]

\newtheorem{theorem}{Theorem}[section]

\theoremstyle{myremark}

\newcommand{\B}{\mathsf{B}}

\newcommand{\R}{\mathbb{R}}

\newcommand{\cD}{\mathcal{D}} 
\newcommand{\cB}{\mathbb{B}}

\newcommand{\cI}{\mathcal{I}}

\newcommand{\cK}{\mathcal{K}}  
\newcommand{\cL}{\mathcal{L}}

\newcommand{\supp}{\operatorname{supp}}

\newcommand{\cl}{\texttt{I}}
\newcommand{\cm}{\texttt{II}}

\newcommand{\dor}{\texttt{d}}

\newcommand{\I}{\{\cl,\cm\}}

\newcommand{\Epar}{\nu}
\newcommand{\Ypar}{\omega}

\newcommand{\EEcut}{\mathcal{V}}
\newcommand{\YYcut}{\mathcal{W}}

\newcommand{\covx}{\mathbf{x}}
\newcommand{\arc}{\phi}

\newcommand*{\email}[1]{\url{#1}}

\numberwithin{equation}{section}
\numberwithin{figure}{section}
\numberwithin{table}{section}

\renewcommand{\cite}{\citet}

\hypersetup{colorlinks=true, linkcolor=blue, citecolor=blue} 

\newcommand{\cmap}{\mathcal{Z}}

\newcommand{\one}{\mathbf{1}}

\newcommand{\interior}[1]{{\kern0pt#1}^{\mathrm{o}}}

\newcommand{\Covx}{\mathbf{X}}
\newcommand{\cOE}{\mathcal{O}_{1}}
\newcommand{\cOY}{\mathcal{O}_{0}}

\newcommand{\cSE}{{h}_{1}}
\newcommand{\cSY}{{h}_{0}}

\usetikzlibrary{shapes.geometric, arrows,positioning}
\usetikzlibrary{arrows}

\definecolor{cornellred}{RGB}{179,27,27} 

\begin{document}
	\setstretch{1.2}
\begin{titlepage}
\title{Risk Preference Types, \\
Limited Consideration, and Welfare\thanks{We thank the editor, Ivan Canay, two anonymous reviewers, Matias Cattaneo, Cristina Gualdani, Elisabeth Honka, Xinwei Ma, Yusufcan Masatlioglu, Julie Mortimer, Deborah Doukas, Roberta Olivieri, and conference participants at FUR22 and at the JBES session at the ESWM23 for helpful comments. Financial support from NSF grants SES-1824448 and SES-2149374 is gratefully acknowledged.
The authors report there are no competing interests to declare.}}
\author{
Levon Barseghyan \\ \small{Department of Economics} \\ \small{Cornell University} \\ \small{\email{lb247@cornell.edu}}
\and
Francesca Molinari \\ \small{Department of Economics} \\ \small{Cornell University} \\ \small{\email{fm72@cornell.edu}}
}

\date{\normalsize{July 10, 2023}}
\maketitle
\begin{abstract}
\noindent
We provide sufficient conditions for semi-nonparametric point identification of a mixture model of decision making under risk, when agents make choices in multiple lines of insurance coverage (\emph{contexts}) by purchasing a \emph{bundle}.
As a first departure from the related literature, the model allows for two preference types.
In the first one, agents behave according to standard expected utility theory with CARA Bernoulli utility function, with an agent-specific coefficient of absolute risk aversion whose distribution is left completely unspecified.
In the other, agents behave according to the dual theory of choice under risk \citep{Yaari1987} combined with a one-parameter family distortion function, where the parameter is agent-specific and is drawn from a distribution that is left completely unspecified.
Within each preference type, the model allows for unobserved heterogeneity in consideration sets, where the latter form at the bundle level -- a second departure from the related literature. 
Our point identification result rests on observing sufficient variation in covariates across contexts, without requiring any independent variation across alternatives within a single context. 
We estimate the model on data on households' deductible choices in two lines of property insurance, and use the results to assess the welfare implications of a hypothetical market intervention where the two lines of insurance are combined into a single one.
We study the role of limited consideration in mediating the welfare effects of such intervention.
\end{abstract}
\setcounter{page}{0}
\thispagestyle{empty}
\end{titlepage}

\section{Introduction}
This paper is concerned with providing sufficient conditions for semi-nonparametric point identification of risk preferences from observation of agents' choices in property insurance markets, and with assessing the welfare impact of policy interventions in these markets. 
Property insurance includes a collection of lines of coverage (e.g., for automobiles: collision, comprehensive, liability, etc.), and we refer to each of them as a \emph{context}.
Within each context, a finite set of alternatives is offered for purchase.
Researchers frequently observe agents choosing (at the same time) one alternative in each context, hence choosing a \emph{bundle}.
We assume that agents choose bundles based on preferences that are stable across contexts (i.e., a single agent-specific parameterization of the model governs that agent's choices in each context).\footnote{This assumption is sometimes viewed as an aspect of rationality \citep[e.g.,][]{Kahneman2003}, and is credible in our empirical study of demand in very similar contexts (collision and comprehensive deductible insurance).} 
Our model allows for unobserved heterogeneity in preference types, with some agents behaving according to expected utility theory with CARA Bernoulli utility function (EU types), and others behaving according to the dual theory of choice under risk \citep{Yaari1987} combined with a one-parameter family distortion function (DT types).
The coefficient of risk aversion of the EU types, and the parameter of the distortion function of the DT types, are random coefficients with unknown distribution functions that are left completely unspecified.
The model also allows for unobserved heterogeneity in the bundles that agents consider before making a choice (their \emph{consideration set}).
In particular, whether an alternative offered in one context is considered can depend in unrestricted ways on whether another alternative offered in a distinct context is also considered.

Such rich unobserved heterogeneity makes identification analysis challenging.
The multiple preference types, random coefficients within type, and agent-specific consideration sets, contribute three layers to a mixtures problem that we need to disentangle.
Moreover, 
because we allow the consideration sets to form at the bundle level, even within a single preference type the choice problem does not inherit the standard \emph{single crossing property} of \citet{mirrlees1971exploration} and \citet{spence1974market}, central to important studies of decision making under risk \citep[e.g.,][]{apesteguia2017single,chiappori2019aggregate}, that \citet{BaMoTh21} show plays a key role in allowing for semi-nonparametric point identification of single preference type models.
Our main methodological contribution amounts to showing how to resolve each of these challenges.
In doing so, we also confront the fact that due to the structure of insurance markets and of data resulting from a single insurance company, while the covariates $\covx$ characterizing products in each context do exhibit independent variation across contexts, they do not exhibit independent variation across alternatives within a context.\footnote{Within a single insurance company, typically in a given context if an agent faces a larger price than another agent for one alternative, the first agent faces a (proportionally) larger price for all other alternatives.}

One may wonder whether some aspects of unobserved heterogeneity that are present in our model could be dispensed with, thereby simplifying the identification problem.
We argue that this is not the case, both in our empirical application and more broadly.
A large literature in experimental economics documents that while some people exhibit behavior consistent with standard EU theory, others exhibit behavior that systematically deviates from it \citep[e.g.,][]{Starmer2000}.
And it  reports substantial heterogeneity in risk preferences within type; see, e.g., \cite{Choi2007} and references therein. 
Moreover, people routinely make or stick to sub-optimal choices \citep{Handel2013,Bhargava2015,BMT16}, or make choices across contexts that imply incompatible levels of risk aversion \citep{Barseghyan2011,Einav2012}. 
The traditional additive error random utility model (Luce-McFadden model), or a  ``trembling hand" alternative \citep[reviewed in][]{Wilcox08} that is sometimes used to study insurance demand, often do not remedy the problem, as the model implied choice probabilities can be incompatible with their empirical counterpart. 
These incompatibilities are not specific to a particular utility model, but to an entire class of models that satisfy properties that are typically viewed as desirable.\footnote{See \cite{BaMoTh21} for a formal discussion and Section \ref{subsection:why_unobs_C} below for further details.} 
On the other hand, models of decision making under risk with limited consideration can rationalize agents' choices.

We illustrate the relevance of the rich unobserved heterogeneity that we allow for, by estimating risk preferences from data on household's choices in two contexts, auto collision and auto comprehensive.
While currently U.S. property insurance companies offer these two lines of coverage as two separate products, we investigate the implications of offering a combined auto insurance product at a price that equals the sum of the prices for the two separate coverages.
Such pricing arises if firms operate under perfect competition or if they use a constant markup rule.
This counterfactual exercise is of substantive interest as combined lines of coverage already exist elsewhere (e.g., in Israel; see \cite{Cohen2007}) and even in the U.S. auto insurance industry.
For example, property damage and bodily injury coverage can be offered both as separate lines of coverage, as well as combined in the form of single limit liability coverage. 
The exercise has the virtue of illustrating the potentially different predictions of the EU model and of the DT model, as we explain in Section \ref{sec:welfare_results}, and the extent to which these predictions interact with whether consideration increases or decreases after the intervention.
Moreover, the exercise informs the debate on the need to simplify insurance choice, and it clarifies how limited consideration interacts with the behavioral responses associated with this type of market intervention.

The rest of the paper is organized as follows.
Section \ref{sec:general:model} lays out the model, using our application as motivating example.
Section \ref{sec:identification} presents our sufficient conditions for its semi-nonparametric point identification.
Section \ref{sec:empirical} describes our empirical model and the data.
Section \ref{sec:results_estimation} reports the results of our estimation exercise.
Section \ref{sec:welfare_results} reports the results of the welfare exercise.
Section \ref{sec:discussion} concludes by contextualizing our work in the broader literature.

\section{Discrete Choice Under Risk in Multiple Contexts}\label{sec:general:model}
Our starting point is the random utility model in \citet{McFadden1974}, applied to study choices over risky alternatives with monetary outcomes.
We further adapt the model to analyze the behavior of agents who make choices under risk in multiple distinct contexts.

\subsection{Lotteries as objects of choice in property insurance}
\label{subsection:lotteries}
We use our empirical application as motivating example for the discrete choice framework that we analyze.
We study deductible choices in two contexts: auto collision (context $\cl$) and auto comprehensive (context $\cm$).
In each context $j=\cl,\cm$, we assume that there are two states of the world: one that has probability $\mu_i^j$, where an accident happens and agent $i$ faces a loss; and the other that has probability $1-\mu_i^j$, where no accident happens.
Auto collision coverage can be used to insure against loss in context $\cl$: it pays for damage in excess of the deductible to the insured vehicle caused by a collision with another vehicle or object, without regard to fault. 
Auto comprehensive coverage can be used to insure against loss in context $\cm$: it pays for damage in excess of the deductible to the insured vehicle from all other causes (e.g., theft, fire, flood, windstorm, or vandalism), without regard to fault. 
In each context, a finite set $\cD^j$ of alternatives (insurance contracts) is offered.

Conditional on risk level, i.e., given $\mu_i^j$, each alternative $\ell \in \cD^j$ is fully characterized by the pair $(\dor^{\ell j},\covx_i^{\ell j})$. 
The first element is the insurance deductible, which is the agent's out of pocket expense if a loss occurs. 
All deductibles are assumed to be less than the lowest realization of the loss and $\dor^{1j}>\dor^{2j}>\dots>\dor^{M^j j}$, with $M^j$ the total number of deductibles in context $j$. 
In collision, $M^\cl=5$ and $\dor^\cl\in\{\$1000,\$500,\$250,\$200,\$100\}$; in comprehensive, $M^\cm=6$ and $\dor^\cm\in\{\$1000,\$500,\$250,\$200,\$100,\$50\}$, for a total of $30$ bundles of offered coverages in $\cD=\cD^\cl\times\cD^\cm$.

The second element in $(\dor^{\ell j},\covx_i^{\ell j})$ is the price (insurance premium), and varies across agents.
It is important to understand the sources of such variation, because to obtain our point identification result we assume that premiums are exogenous to preferences (Assumption \ref{ass:coeff_restrict} below) and exhibit substantial variation across households (Assumptions \ref{ass:distinct_contexts}-\ref{ass:var_x} below).

First, we note that an insurance company's rating plan is subject to state regulation and oversight. 
In particular, the regulations require that a company receive prior approval of its rating plan by the state insurance commissioner, and they prohibit the company and its agents from charging rates that depart from the plan. 

Second, we describe the procedure applied by the company from which we obtained our data to rate a policy in each line of coverage.\footnote{See Section \ref{sec:data} below for additional information on the data.} 
Under the plan, within each context $j$ the company determines a household's base price $\covx_i^j$ according to a coverage-specific rating function, which takes into account agent $i$'s coverage-relevant characteristics and any applicable discounts. 
Using the base price, the company then generates the agent's pricing menu $\mathcal{M}^j\equiv\{(\dor^{\ell j},\covx_i^{\ell j}):\ell\in\cD^j\}$, which associates a premium $\covx_i^{\ell j}$ with each deductible $\dor^{\ell j}$ in the coverage-specific set of alternatives in $\cD^j$, according to an agent-invariant and coverage-specific multiplication rule, $\covx_i^{\ell j}=(g^{\ell j}\cdot\covx_i^j)+\delta^j$, where $\delta^j>0$ and $g^{\ell j}$ is increasing in $\ell$ and strictly greater than zero, so that $\covx_i^{1j}<\covx_i^{2j}<\dots<\covx_i^{M^j j}$ ($\dor^{\ell j}$ is decreasing in $\ell$, so lower deductibles provide more coverage and cost more).\footnote{The multiplicative factors $\{g^{\ell j}:\ell\in\cD^j\}$ are known as the deductible factors and $\delta^j$ is a small markup known as the expense fee.}
As $\{g^{\ell j}:\ell\in\cD^j;\delta^j\}$ are agent-invariant, there is \emph{no} independent variation in covariates across alternatives within a context.

With this as background, for given $\mu_i^j$, alternatives can be represented as lotteries:
\begin{align}
\cL(\dor^{\ell j},\covx_i^j)\equiv \left( -\covx_i^{\ell j},1-\mu_i^j ;-\covx_i^{\ell j}-\dor^{\ell j},\mu_i^j \right),\label{eq:lottery}
\end{align} 
where  $(\covx_i^j,\mu_i^j)$ is observed by the researcher for each agent $i$ and context $j$. 
Throughout, we implicitly condition on $\mu_i^j$.
We do not use variation in $\mu_i^j$ to establish our identification results, although doing so is potentially useful and the subject of ongoing research.

\subsection{Preference types with unobserved heterogeneity within type}
\label{subsection:types}
We allow the population of agents to be a \emph{mixture of preference types}.\footnote{Multiple preference types are a focus of the literature that estimates risk preferences using experimental data (e.g., \cite{Bruhin2010,Conte2011}; \cite{Harrison2010}), although preferences are homogeneous within each type, at most conditioning on some observed demographic characteristics.}
The literature has put forward many models of decision making under risk which can generate demand for insurance at actuarially unfair prices, including the workhorse expected utility theory model and a host of non-expected utility theory models.
Each of these models has relative (de)merits in rationalizing observed choices, and may deliver different predictions for counterfactual policies \citep[see][for a review]{Barseghyan2018}.
We hence think it important to provide identification results for a model where multiple preference types are allowed for, and where unobserved heterogeneity within type is also present.

For notational simplicity, we detail here the case with two preference types.
The results extend to more than two types (even when one observes choices only in two contexts).
Let each agent $i$ draw a preference type $t_i$ as follows:
\begin{align}
t_{i}=\left\{
\begin{tabular}
[c]{ll}
$1$ & with probability $\alpha$,\\
$0$ & with probability $1-\alpha$,
\end{tabular}
\right.  \label{eq:t_i}
\end{align}
with $\alpha\in(0,1)$ the unknown mixing probability.

Each realization of $t_i$ is associated with a family of utility functions with \emph{distinct functional forms}, denoted $\mathcal{U}^1=\{U_{\Epar},~\Epar\in [0,\bar \Epar]\}$ for $t_i=1$, and $\mathcal{U}^0=\{U_{\Ypar},~\Ypar\in[0,\bar \Ypar]\}$ for $t_i=0$.
Functions in each family are known up to a scalar random coefficient that depends on type, denoted $\Epar_i$ (with support $[0,\bar \Epar]$) for agents with $t_i=1$, and $\Ypar_i$ (with support $[0,\bar \Ypar]$) for agents with $t_i=0$.
For example,  in our empirical application $\mathcal{U}^1$ is the collection of expected utility functions associated with preferences that exhibit constant absolute risk aversion (CARA) with agent-specific Arrow-Pratt coefficient $\Epar_i$,\footnote{Other preferences that are characterized by a scalar parameter include ones exhibiting constant relative risk aversion (CRRA), or negligible third derivative \citep[NTD; see, e.g.,][]{Cohen2007,Barseghyan2013}.  Under CRRA, it is required that agents' initial wealth is known to the researcher. } and $\mathcal{U}^0$ is a family of non-expected utility functions that do not nest expected utility as a special case and are parametrized by $\Ypar_i$ (see Eqs.~\eqref{eq:EUT}-\eqref{eq:Yaari} and Assumptions \ref{ass:cara} \& \ref{ass:Omega} below).
As the preference types are distinct, each agent either receives a draw of $\Epar_i$ or a draw of $\Ypar_i$, hence by construction the two random coefficients are independent.
We do not impose \emph{any} parametric restrictions on their distributions.
Rather, in Section \ref{sec:identification} we provide nonparametric point identification results for the two marginal distributions of preferences and for the share of each type.
\begin{assumption}[Restrictions on distribution of random coefficients]\label{ass:coeff_restrict}
The random coefficient $\Epar_i$ (respectively, $\Ypar_i$) is distributed according to a cumulative distribution function $F$ (respectively, $G$) that satisfies the properties of CDFs,  and admits a density function $f$ that is continuous and strictly positive on $\mathcal{V}\equiv[0,\bar \Epar]$ (respectively,  $g$ strictly positive on $\mathcal{W}\equiv[0,\bar \Ypar]$).
Both $\Epar_i$ and $\Ypar_i$ are independent of $\covx_i$.\footnote{Recall that our analysis conditions on $\mu_i^j$, hence the distribution of preferences may depend on it.}
\end{assumption}

We make three fundamental assumptions about utility functions in both families.
First, we assume that households' preferences are \emph{stable} across contexts, which allows us to leverage variation in observed choices and covariates across contexts (recall that we have no covariate variation within each context).
\begin{assumption}[Stability]
\label{ass:stability}
The utility function $U_{\Epar_i}$ of each agent $i$ with $t_i=1$ (respectively, $U_{\Ypar_i}$ for agents with $t_i=0$) is \emph{context-invariant}.
\end{assumption}

Second, we need to take a stand on how agents make choices in multiple contexts.
To this end, it is important to introduce notation for \emph{bundles} of alternatives.
Denote bundles as $\cI_{\ell,q}$, where the first index refers to the alternative in context $\cl$ and the second one to that in context $\cm$.
Let $CE_{\Epar_i}(\cL(\dor^{\ell j},\covx^j))$ (respectively, $CE_{\Ypar_i}(\cL(\dor^{\ell j},\covx^j))$) denote the \emph{certainty equivalent} of lottery $\cL(\dor^{\ell j},\covx^j)$ \citep[see, e.g.,][Definition 6.C.2]{mas:whi:gre95} in context $j$ for an agent of type $t_i=1$ (respectively, $t_i=0$).
We impose a standard, albeit sometimes implicit, assumption in the literature,\footnote{All papers that estimate risk preferences in the field as reviewed in \cite{Barseghyan2018} impose it. } according to which agents' choices are made without taking into account any background risk \citep[e.g.,][]{Read1999}.
\begin{assumption}[Narrow Bracketing]
\label{ass:narrow}
Agent $i$'s certainty equivalent for the lottery associated with bundle $\cI_{\ell,q}$ is equal to $CE_{\zeta_i}(\cL(\dor^{\ell \cl},\covx^\cl))+CE_{\zeta_i}(\cL(\dor^{q \cm},\covx^\cm))$, with $\zeta_i=\Epar_i$ if $t_i=1$ and $\zeta_i=\Ypar_i$ if $t_i=0$.
\end{assumption}

Third, we assume that each preference type satisfies the classic Single Crossing Property (SCP) of \citet{mirrlees1971exploration} and \citet{spence1974market}, central to important studies of decision making under risk \citep[see, for example][]{apesteguia2017single,chiappori2019aggregate}.\footnote{The SCP is satisfied in many contexts, ranging from single agent models with goods that can be unambiguously ordered based on quality, to multiple agents models \citep[e.g.,][]{athey2001single}.} 
Formally,
\begin{assumption}[Single Crossing Property]\label{ass:SCP} For a given context $j$ and any two lotteries $\cL(\dor^{\ell j},\covx)$ and $\cL(\dor^{k  j},\covx)$, $\ell<k$, there exists a continuously differentiable and strictly monotone function $\cmap_k^\ell: \supp(\covx) \to \R_{[-\infty,\infty]}$, with $\supp(\covx)=\bigcup_{j=\cl,\cm}\supp(\covx^j)$, such that
	\begin{align*}
	U_\zeta(\cL(\dor^{k  j},\covx)) &< U_\zeta(\cL(\dor^{\ell j},\covx))  \quad \forall \zeta \in (-\infty,\cmap_k^\ell(\covx)), \\
	U_\zeta(\cL(\dor^{k  j},\covx)) &= U_\zeta(\cL(\dor^{\ell j},\covx))  \quad ~\zeta = \cmap_k^\ell(\covx), \\
	U_\zeta(\cL(\dor^{k  j},\covx))&> U_\zeta(\cL(\dor^{\ell j},\covx))  \quad \forall \zeta \in (\cmap_k^\ell(\covx),\infty).
	\end{align*}
	where $\zeta=\Epar_i$ for agents of type $t_i=1$ and $\zeta=\Ypar_i$ for type $t_i=0$.
	We refer to $\cmap_k^\ell(\cdot)$ as the cutoff between $\cL(\dor^{\ell j},\covx)$ and $\cL(\dor^{k j},\covx)$, and denote it $\EEcut_k^\ell(\cdot)$ for $t_i=1$ and $\YYcut_k^\ell(\cdot)$ for $t_i=0$.\footnote{We assume that while $\Epar$ and $\Ypar$ have bounded support, the utility functions in $\mathcal{U}^1$ and $\mathcal{U}^0$ are well defined for any real valued $\Epar$ and $\Ypar$, respectively.}
\end{assumption}
Within a single context, the expected utility theory framework generally satisfies the SCP, which requires that if an agent with a certain degree of risk aversion (the random coefficient $\Epar_i$) prefers a safer lottery to a riskier one, then all agents with higher risk aversion also prefer the safer lottery.\footnote{For a discussion of possible failures of SCP, see \citet{ape:bal18}. }
The same is true for the non-expected utility theory model that we use in our empirical analysis in Section \ref{sec:empirical}.
The SCP implies that \emph{within a single context}, the household's ranking of alternatives is monotone in $\Epar_i$ for $t_i=1$ and in $\Ypar_i$ for $t_i=0$, yielding vertical differentiation of alternatives within each preference type.  

\subsection{Unobserved heterogeneity in consideration sets}
\label{subsection:consideration}
Across contexts, the subset of alternatives actually available to each agent is unknown to the researcher, due, e.g., to unobserved budget constraints, liquidity constraints, etc.
Moreover, agents face an overall large and potentially overwhelming universe of feasible alternatives, leading to choice overload, cognitive ability constraints, etc.
Hence, we allow for \emph{unobserved heterogeneity in consideration sets}, i.e., in the collection of alternatives that the agents evaluate when making their choices.
We denote the overall universe of alternatives across contexts as
$\cD\equiv\cD^\cl\times\cD^\cm$, with $\cI_{\ell,q}$ denoting each of the bundles in $\cD$.
\begin{assumption}[Consideration set formation mechanism]\label{ass:consideration}
Conditional on $t_i$, agent $i$ draws a consideration set $C_i\subseteq\cD$ independently from its random coefficient and from $\covx_i$ s.t.
\begin{align*}
\mathcal{Q}_{1}(\cK)&\equiv\Pr(C_i=\cK|t_i=1)=\Pr(C_i=\cK|\covx_i,\Epar_i,t_i=1),~~~\cK\subseteq \cD,\\
\mathcal{Q}_{0}(\cK)&\equiv\Pr(C_i=\cK|t_i=0)=\Pr(C_i=\cK|\covx_i,\Ypar_i, t_i=0),~~~\cK\subseteq \cD.
\end{align*}
\end{assumption}
The fundamental restrictions imposed in Assumption \ref{ass:consideration} are that conditional on preference type, consideration is independent of the agent's random coefficient and of the observed covariate $\covx$.\footnote{Recall that our analysis conditions on $\mu_i^j$, hence the distribution of consideration sets may depend on it.}
However, the distribution of consideration sets may depend on preference type.
Importantly, we allow consideration to be \emph{broad}, as it is determined at the bundle level instead of within context.
A significantly more restrictive approach would posit that consideration is \emph{narrow}: agent $i$ draws a pair of consideration sets 
$C_i^j\in\cD^j$, $j=\cl,\cm$ independently across contexts, and forms $C_i=C_i^\cl\times C_i^\cm$.
As we further discuss in Section \ref{subsec:optimal_choice}, allowing consideration sets to be drawn at the bundle level substantially complicates the identification analysis, but delivers a more realistic model.

\subsection{Optimal choice within the consideration set}
\label{subsec:optimal_choice}
Once the consideration set is drawn, each agent chooses the best alternative in each context according to their preferences.
\begin{align}
\cI^\ast\equiv[\ell^*,q^*]=\arg\underset{[\ell,q]\in C}{\max}CE_{\zeta}(\cL(\dor^{\ell\cl},\covx^\cl))+CE_{\zeta}(\cL(\dor^{q \cm},\covx^\cm)),\label{eq:max_CE}
\end{align}
where $\zeta=\Epar$ if $t=1$ and $\zeta=\Ypar$ if $t=0$.
The bundle choice $\cI^*$ depends on the agent's preference type, random coefficient, consideration set, associated premium-deductible tuples, and claim probabilities $\mu^j,~j=\cl,\cm$.

The flexible model of consideration set formation that we allow for has important implications for the choice problem in Eq.~\eqref{eq:max_CE}.
If we were to assume narrow consideration, hence restrict agents to draw consideration sets \emph{independently} across contexts, the choice problems would break into independent, context-specific decisions, with
\begin{align}
\ell^{*j}=\arg\underset{\ell\in C^j}{\max} ~CE_{\zeta}(\cL(\dor^{\ell^j},\covx^j)).\label{eq:choice_narrow_C}
\end{align}
Each context-specific choice problem satisfies the SCP in Assumption \ref{ass:SCP}.
\citet{BaMoTh21} offer a comprehensive analysis of the implications of the SCP for semi-nonparametric identification of a model of discrete choice under risk that features a single preference type and unobserved heterogeneity in consideration sets. 
Even in the simplified framework where consideration is narrow, our analysis extends theirs as we allow for multiple preference types.
More importantly, a narrow model of consideration implies that very similar alternatives in different contexts enter the consideration set independently.\footnote{For example,  a \$500 deductible at price $\covx^\cl$ in collision insurance and a \$500 deductible at price $\covx^\cm$ in comprehensive insurance would enter the consideration set independently.}
This assumption is unpalatable, particularly when analyzing demand for bundled products.
We therefore allow for broad consideration.
In doing so, we overcome a substantial hurdle relative to \citet{BaMoTh21}.
When consideration is broad and $C_i$ is formed at the bundle level, the SCP may not necessarily hold across tuples of alternatives, because alternatives may not be monotonically ranked against each other (with respect to $\Epar_i$ or $\Ypar_i$). 
Hence, here we develop a new approach to obtain point identification of the distribution of preferences, shares of preferences types, and features of the distribution of consideration sets given type. 

\section{Identification Results}\label{sec:identification}
We begin by describing the conditions under which we can prove our point identification results.\footnote{The results extend easily to more than two contexts, at the cost of heavier notation.}
We index bundles as $\cI_{\ell,q}$ and $\cI_{k,r}$, with $\ell,k\in\cD^\cl$ alternatives in context $\cl$ and $q,r\in\cD^\cm$ alternatives in context $\cm$.
We recall that in each context, $\dor^{1j}>\dor^{2j}>\dots>\dor^{M^j j}$ and $\covx_i^{1j}<\covx_i^{2j}<\dots<\covx_i^{M_j j}$, see Section \ref{subsection:lotteries}.
We let $\EEcut^{\ell,q}_{k,r}(\covx)$ and $\YYcut^{\ell,q}_{k,r}(\covx)$ denote cutoff levels for $\Epar_i$ and $\Ypar_i$, respectively, at which the agent is indifferent between bundles $\cI_{\ell,q}$ and $\cI_{k,r}$.
Hence, under Assumption \ref{ass:narrow}, the cutoff $\EEcut^{\ell,q}_{k,r}(\covx)$ is such that (and similarly for $\YYcut^{\ell,q}_{k,r}(\covx)$):
\begin{multline}
CE_{\EEcut^{\ell,q}_{k,r}(\covx)}(\cL(\dor^{\ell \cl},\covx^\cl))+CE_{\EEcut^{\ell,q}_{k,r}(\covx)}(\cL(\dor^{q \cm},\covx^\cm))\\
=CE_{\EEcut^{\ell,q}_{k,r}(\covx)}(\cL(\dor^{k \cl},\covx^\cl))+CE_{\EEcut^{\ell,q}_{k,r}(\covx)}(\cL(\dor^{r \cm},\covx^\cm)).\label{eq:def_EEcut_CE}
\end{multline}
Relative to the cutoffs introduced in Assumption \ref{ass:SCP}, which compared alternatives within a single context and we denoted $\EEcut^\ell_k(\covx)$ (single superscript and subscript for a single context of choice), we have $\EEcut^{m,q}_{m,r}(\covx)=\EEcut^q_r(\covx)$ and $\EEcut^{\ell,s}_{k,s}(\covx)=\EEcut^\ell_k(\covx)$ for all $m\in\cD^\cl$ and $s\in\cD^\cm$ (and similarly for $\YYcut^{\cdot,\cdot}_{\cdot,\cdot}(\covx)$).
While cutoffs $\EEcut^{\ell,q}_{k,r}(\covx)$ and $\YYcut^{\ell,q}_{k,r}(\covx)$ for $\ell\neq k,q\neq r$ depend on both $\covx^\cl$ and $\covx^\cm$, cutoffs $\EEcut^{\ell,s}_{k,s}(\covx)$ and $\YYcut^{\ell,s}_{k,s}(\covx)$ depend \emph{only} on $\covx^\cl$, while cutoffs $\EEcut^{m,q}_{m,r}(\covx)$ and $\YYcut^{m,q}_{m,r}(\covx)$ depend \emph{only} on $\covx^\cm$.
These properties will be used to establish our identification results.

We remark that the cutoffs $\EEcut^{\ell,q}_{k,r}(\covx)$ and $\YYcut^{\ell,q}_{k,r}(\covx)$ may not be unique if $\ell>k$ but $q<r$ (or vice versa). However, they are unique whenever $\cI_{1,1}$ is compared with \emph{any} other bundle (and similarly whenever $\cI_{M^\cl,M^\cm}$ is compared with any other bundle).

Throughout, we assume that the researcher has access to data that identify the joint distribution of chosen bundles and covariates. 
The consideration set, however, is not observed.
\begin{assumption}[Observed data]
\label{ass:DGP}A random sample $\{(\cI^\ast_i,\covx_i^\cl,\covx_i^\cm):i=1,\dots,n\}$ is observed, with $\cI^\ast_i$, as defined in Eq.~\eqref{eq:max_CE}.
\end{assumption}

\subsection{Restrictions on variation in $\covx^j$ across contexts}
\label{subsec:ass_x_model}
Identification of the model's functionals rests on the interplay between the model and the variation in the observed covariates.
We only require the covariates $\covx_i\equiv(\covx_i^\cl,\covx_i^\cm)$ to vary across agents and contexts, as formally stated below, but allow $\covx_i^\cl$ (respectively, $\covx_i^\cm$) to be constant across alternatives within $\cD^\cl$ (respectively, $\cD^\cm$).
Hence, one needs sufficient variation across contexts to obtain point identification results.
\begin{assumption}[Preferred within a triplet]\label{ass:not_dominated}
In each context $j\in\I$, for any $\covx$ and triplet $\{\dor^{1j},\dor^{kj},\dor^{(k+1)j}\}$, $\forall k\in \{2,...,M^j-1\}$, there are values of $\Epar$ (and $\Ypar$) at which each alternative in this triplet is strictly preferred to the other two.   
\end{assumption}
Assumption \ref{ass:not_dominated} requires that given three coverage levels including the cheapest, each one is preferred by at least some agent.
As shown in \citet{BaMoTh21}, under Assumption \ref{ass:SCP}, this condition is satisfied for agents of type $t_i=1$ within context $\cl$ if and only if $-\infty<\EEcut^{1,1}_{2,1}(\covx^\cl)<\EEcut^{1,1}_{3,1}(\covx^\cl)<\EEcut^{1,1}_{4,1}(\covx^\cl)\dots<+\infty$ (and similarly for agents of type $t_i=0$, and for context $\cm$ with appropriate modifications in the compared bundles and evaluation at $\covx^\cm$ instead of $\covx^\cl$), with $\EEcut^{\ell,q}_{k,r}$ defined through Eq.~\eqref{eq:def_EEcut_CE}. 
So, any agent of type $t_i=1$ who draws $\Epar<\EEcut^{1,1}_{2,1}(\covx^\cl)$ unambiguously prefers alternative $\ell^{1 \cl}$ to any other in $\cD^\cl$.

In what follows, an important role is played by the values of $\covx=(\covx^\cl,\covx^\cm)$ at which the indifference cutoff for an agent of type $t_i$ between alternatives $\ell^{1 \cl}$ and $\ell^{2 \cl}$ (the two cheapest alternatives in context $\cl$) is equal to that agent's indifference cutoff between alternatives $\ell^{1 \cm}$ and $\ell^{2 \cm}$ (the two cheapest alternatives in context $\cm$).
We first define these values of $\covx$, and then make assumptions on the support of $\covx$ to guarantee that it includes them.
\begin{definition}[Covariate values delivering indifference]\label{def:indifferenceX}
Given $t_i$, fix a value of $\Epar\in[0,\bar\Epar]$ if $t_i=1$ and of $\Ypar\in[0,\bar\Ypar]$ if $t_i=0$.
Let the set of covariate values at which the agent has preference $\Epar$ (respectively,  $\Ypar$) and is indifferent between bundles $\cI_{1,1}$, $\cI_{1,2}$, and $\cI_{2,1}$, be: 
	\begin{align*} 
	 \Covx^1(\Epar)&\equiv\{(\covx^\cl,\covx^\cm):  \EEcut^{1,1}_{2,1}(\covx^\cl)=\EEcut^{1,1}_{1,2}(\covx^\cm)=\Epar\},\\
	 \Covx^0(\Ypar)&\equiv\{(\covx^\cl,\covx^\cm):  \YYcut^{1,1}_{2,1}(\covx^\cl)=\YYcut^{1,1}_{1,2}(\covx^\cm)=\Ypar\}.
	\end{align*}
	\end{definition}
The covariate values $\Covx^1(\Epar)$ (respectively, $\Covx^0(\Ypar)$) are the values of $\covx=(\covx^\cl,\covx^\cm)$ at which an agent with preferences $\Epar$ (respectively, $\Ypar$) is indifferent between the two cheapest coverage levels in context $\cl$ and, at the same time, also in context $\cm$. 
In other words, the agent is indifferent between $\cI_{1,1}$, $\cI_{1,2}$, and $\cI_{2,1}$ (and, hence, $\cI_{2,2}$).
Given the single crossing property in Assumption \ref{ass:SCP}, within each context it is immediate to see that both elements of $\Covx^1(\Epar)$ (the covariate value in context $\cl$ and the covariate value in context $\cm$) are strictly monotone in $\Epar$ (and, similarly, both elements of $\Covx^0(\Ypar)$ are monotone in $\Ypar$).
For example, the higher is $\Epar$, the higher is the base price in context $\cl$ at which the agent with random coefficient $\Epar$ is indifferent between $\cI_{1,1}$ and $\cI_{2,1}$.
Hence, we can represent $\Covx^1(\Epar)$ (respectively, $\Covx^0(\Ypar)$) as a strictly monotone function on the support of $(\covx^\cl,\covx^\cm)$.\footnote{See Figure \ref{Fig:boldX} and its discussion below.}
We assume that these strictly monotone functions intersect on a set of measure zero.
\begin{assumption}[Distinct contexts]\label{ass:distinct_contexts}
The contexts are distinct, in the sense that:
\begin{enumerate}[label=(\Roman*)]
\item\label{distinct_minimal} $\Covx^1(\Epar)\neq\Covx^0(\Ypar)~a.e.$
\item\label{distinct_sufficient} The following four conditions are satisfied:
		\begin{align}
		\EEcut^{1,1}_{\ell,q}(\Covx^0(\Ypar))&\neq\EEcut^{1,1}_{k,r}(\Covx^0(\Ypar)) ~a.e.\label{eq:distinct1}\\ 
		\YYcut^{1,1}_{\ell,q}(\Covx^1(\Epar))&\neq\YYcut^{1,1}_{k,r}(\Covx^1(\Epar)) ~a.e. \label{eq:distinct2}\\
		\EEcut^{1,1}_{\ell,q}(\Covx^1(\Epar))&\neq\EEcut^{1,1}_{k,r}(\Covx^1(\Epar)) ~a.e. ~\forall \{\ell,q,k,r\} ~s.t.~\{\ell,q,k,r\}\setminus \{1,2\}\neq\emptyset.\label{eq:distinct3}\\
		\YYcut^{1,1}_{\ell,q}(\Covx^0(\Ypar))&\neq\YYcut^{1,1}_{k,r}(\Covx^0(\Ypar)) ~a.e. ~\forall \{\ell,q,k,r\} ~s.t.~\{\ell,q,k,r\}\setminus \{1,2\}\neq\emptyset.\label{eq:distinct4}
	\end{align}
\end{enumerate}
\end{assumption}
Assumption \ref{ass:distinct_contexts}-\ref{distinct_sufficient} implies Assumption \ref{ass:distinct_contexts}-\ref{distinct_minimal}, as Eqs.~\eqref{eq:distinct1}-\eqref{eq:distinct2} for $\ell,q=2,1$ and $k,r=1,2$ imply $\Covx^1(\Epar)\neq\Covx^0(\Ypar)~a.e.$ 
Both conditions require that at any value of $(\covx^\cl,\covx^\cm)$ at which indifference across $\cI_{1,1}$, $\cI_{1,2}$, and $\cI_{2,1}$ occurs for an agent of type $t_i=1$, such indifference cannot occur for an agent of type $t_i=0$. 
Additionally, Assumption \ref{ass:distinct_contexts}-\ref{distinct_sufficient} requires that at any value of $(\covx^\cl,\covx^\cm)$ at which indifference across $\cI_{1,1}$, $\cI_{1,2}$, and $\cI_{2,1}$ occurs, no other triplet of bundles including $\cI_{1,1}$ can generate a three-way tie in utility ranking. Given the data and utility models across preference types, one can directly check whether Assumption \ref{ass:distinct_contexts} is satisfied. 
\begin{figure}
	\centering
	\adjincludegraphics[scale=.35]{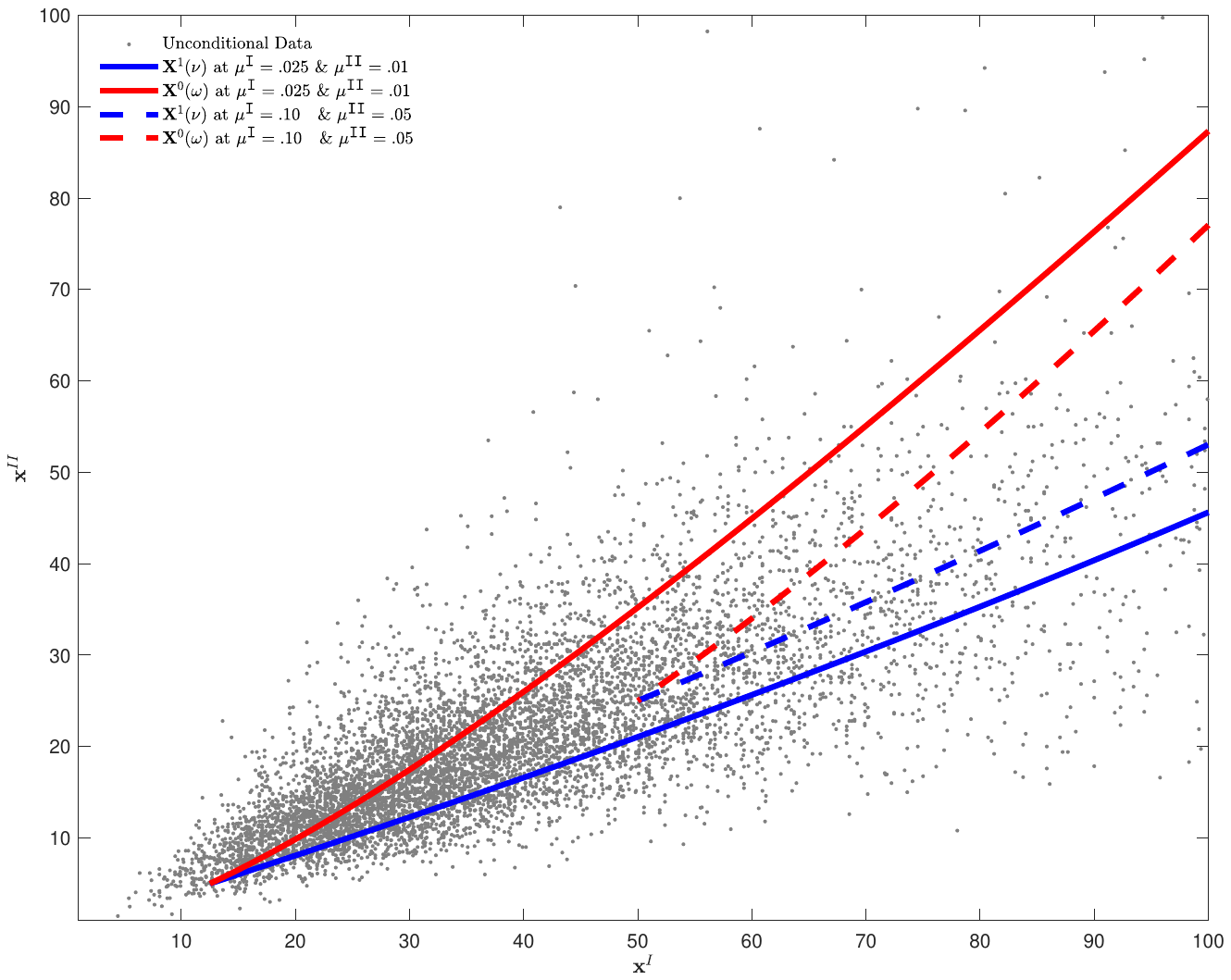}
\caption{\small{$\Covx^0(\Ypar)$ and $\Covx^1(\Epar)$ in our application, with data in the background.}}\label{Fig:boldX}
\end{figure}
Finally, we require that the support of $\covx$ is sufficiently rich, as point identification of $f(\Epar)$ and $g(\Ypar)$ can only occur at values of $\Epar$ and $\Ypar$ that belong, respectively, to intervals $[\Epar^*,\Epar^{**}]\subseteq [0,\bar\Epar]$ and $[\Ypar^*,\Ypar^{**}]\subseteq[0,\bar\Ypar]$ satisfying the next assumption.
\begin{assumption}[Independent variation in $\covx$]\label{ass:var_x}
Let $[\Epar^*,\Epar^{**}]\subseteq [0,\bar\Epar]$ and $[\Ypar^*,\Ypar^{**}]\subseteq[0,\bar\Ypar]$ be intervals such that, for some $\epsilon>0$, the random vector $\covx=(\covx^\cl,\covx^\cm)$ has strictly positive density on the sets $\mathcal{S}^1_\epsilon(\Epar^*,\Epar^{**})\subset \R^2$ and $\mathcal{S}^0_\epsilon(\Ypar^*,\Ypar^{**})\subset \R^2$, with
\begin{align*}
\mathcal{S}^1_\epsilon(\Epar^*,\Epar^{**})&=\left\{\B_\epsilon(\Covx^1(\Epar)),~\Epar\in[\Epar^*,\Epar^{**}]\right\},\\
\mathcal{S}^0_\epsilon(\Ypar^*,\Ypar^{**})&=\left\{\B_\epsilon(\Covx^0(\Ypar)),~\Ypar\in[\Ypar^*,\Ypar^{**}]\right\}.
\end{align*}
where $\B_a(c)$ denotes a ball in $\R^2$ of radius $a$ centered at $c$.
\end{assumption}
Assumption \ref{ass:var_x} guarantees that for each $\Epar\in[\Epar^*,\Epar^{**}]$ there are values of $\covx$ such that $\Covx^1(\Epar)$ is non-empty and that there is an $\epsilon$-neighborhood around $\Covx^1(\Epar)$ with positive density (and similarly for $\Covx^0(\Ypar)$ and all $\Ypar\in[\Ypar^*,\Ypar^{**}]$). This yields sufficient observed variation in $\covx$ to identify the functionals that we are after. 
We illustrate the notion of distinct contexts and independent variation in $\covx$ via Figure \ref{Fig:boldX}, which depicts $\Covx^0(\Ypar)$ and $\Covx^1(\Epar)$ drawn for different pairs of $\mu$'s. First, $\Covx^0(\Ypar)$ and $\Covx^1(\Epar)$ intersect only at a single point.\footnote{In our empirical model described in Section \ref{sec:empirical}, this intersection point corresponds to $\Epar=0$ and $\Ypar=1$, i.e., respectively,  no risk aversion and no probability distortions. }
Second, these curves are both monotone. 
We present them with a scatterplot of unconditional data from our empirical application in the background, to highlight the fact that even when variation in $\covx$ does not cover the entire $\R^2_{+}$, identification is attainable since Assumption \ref{ass:var_x} requires variation in $\covx$ only to cover respective neighborhoods of $\Covx^0(\Ypar)$ and $\Covx^1(\Epar)$. 
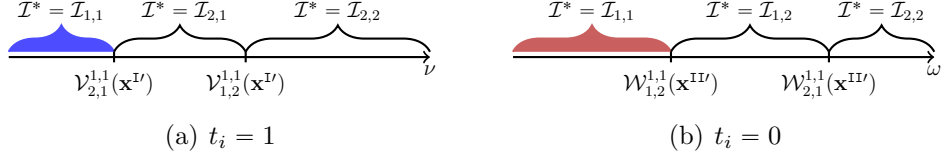
\begin{figure}
\centering
\subfigure[$t_i=1$]{
\begin{tikzpicture}[thick,scale=.7, every node/.style={transform shape}]
\draw[->] (0,0) -- (8,0) node[below] {$\Epar$};
	\draw[-] (2,-.1) -- (2,.1) ;	
	\draw[-] (4.5,-.1) -- (4.5,.1) ;	
	\draw (2,-.1) node[below] {$\EEcut^{1,1}_{2,1}(\covx^{\cl \prime})$};
	\draw (4.5,-.1) node[below] {$\EEcut^{1,1}_{1,2}(\covx^{\cl \prime})$};

	\draw (1,.5) node[above] {$\cI^*=\cI_{1,1}$};
	\draw (3.35,.5) node[above] {$\cI^*=\cI_{2,1}$};
	\draw (6.25,.5) node[above] {$\cI^*=\cI_{2,2}$};
	
	\draw[decoration={brace,amplitude=10pt,raise=2pt},decorate,blue!70,fill] (0,0) -- (2,0);
	\draw[decoration={brace,amplitude=10pt,raise=2pt},decorate] (2,0) -- (4.5,0);
	\draw[decoration={brace,amplitude=10pt,raise=2pt},decorate] (4.5,0) -- (8,0);
\end{tikzpicture}}
\hspace{.5cm}
\subfigure[$t_i=0$]{
\begin{tikzpicture}[thick,scale=.7, every node/.style={transform shape}]
\draw[->] (0,0) -- (8,0) node[below] {$\Ypar$};
	\draw[-] (3,-.1) -- (3,.1) ;	
	\draw[-] (6,-.1) -- (6,.1) ;	
	\draw (6,-.1) node[below] {$\YYcut^{1,1}_{2,1}(\covx^{\cm \prime})$};
	\draw (3,-.1) node[below] {$\YYcut^{1,1}_{1,2}(\covx^{\cm \prime})$};

	\draw (1.5,.5) node[above] {$\cI^*=\cI_{1,1}$};
	\draw (4.5,.5) node[above] {$\cI^*=\cI_{1,2}$};
	\draw (7,.5) node[above] {$\cI^*=\cI_{2,2}$};
	
	\draw[decoration={brace,amplitude=10pt,raise=2pt},decorate,cornellred!70,fill] (0,0) -- (3,0);
	\draw[decoration={brace,amplitude=10pt,raise=2pt},decorate] (3,0) -- (6,0);
	\draw[decoration={brace,amplitude=10pt,raise=2pt},decorate] (6,0) -- (8,0);
\end{tikzpicture}}
\caption{Stylized depiction of regions where, under full consideration, 
$\cI^*=\cI_{1,1}$.}\label{Fig:full}
\end{figure}

We next explain why, under full consideration, our assumptions suffice for identification of the share of preference types and the distributions of the respective random coefficients.
Fix a value of $\Epar\in[\Epar^*,\Epar^{**}]$ at which one wants to learn $f(\Epar)$.
Under Assumption \ref{ass:var_x}, $\Covx^1(\Epar)$ is non-empty and there is an $\epsilon$-ball of positive density around it.
Along with Assumption \ref{ass:distinct_contexts}, this implies that there is a vector $(\covx^{\cl \prime},\covx^{\cm\prime})\in \B_\epsilon(\Covx^1(\Epar))$ such that $\Epar=\EEcut^{1,1}_{2,1}(\covx^{\cl \prime})<\EEcut^{1,1}_{1,2}(\covx^{\cm \prime})$ and $\YYcut^{1,1}_{2,1}(\covx^{\cl \prime})>\YYcut^{1,1}_{1,2}(\covx^{\cm \prime})$. 
Then, as shown in Figure \ref{Fig:full}, under Assumptions \ref{ass:SCP} and \ref{ass:not_dominated},\footnote{Recall that these assumptions, jointly, imply that any agent who draws $\Epar<\EEcut^{1,1}_{2,1}(\covx^{\cl \prime})<\EEcut^{1,1}_{1,2}(\covx^{\cm \prime})$ unambiguously prefers alternative $\ell^{1\cl}$ to all other alternatives in $\cD^\cl$, unambiguously prefers alternative $\ell^{1\cm}$ to all other alternatives in $\cD^\cm$, and therefore unambiguously prefers bundle $\cI_{1,1}$ to any other bundle in $\cD$.} 
\begin{align}
\Pr(\cI^*=\cI_{1,1}|\covx^\prime)&=\alpha F(\EEcut^{1,1}_{2,1}(\covx^{\cl \prime}))+(1-\alpha)G(\YYcut^{1,1}_{1,2}(\covx^{\cm \prime})).\label{eq:choice_prob_full}
\end{align}
In turn, owing to the fact that $\EEcut^{1,1}_{2,1}(\covx^{\cl \prime})$ depends on $\covx^\cl$ but $\YYcut^{1,1}_{1,2}(\covx^{\cm \prime})$ does not, this yields
\begin{align}
\frac{\partial \Pr(\cI^*=\cI_{1,1}|\covx^\prime)}{\partial \covx^{\cl}}=\alpha f(\Epar)\frac{\partial \EEcut^{1,1}_{2,1}(\covx^{\cl \prime})}{\partial \covx^{\cl}},\label{eq:deriv_choice_prob_full}
\end{align}
where the term $\frac{\partial\EEcut^{1,1}_{2,1}(\covx^{\cl \prime})}{\partial \covx^{\cl}}$ is a known function of $\covx^\cl$ and is different from zero due to Assumption \ref{ass:SCP} (where cutoff functions are assumed to be strictly monotone in $\covx$).
If $[\Epar^*,\Epar^{**}]=[0,\bar\Epar]$, one can repeat the above argument for all $\Epar$ on the support and then use the fact that $f(\Epar)$ integrates to one to learn $\alpha$.
One can similarly learn $g(\Ypar)$, $\Ypar\in[\Ypar^*,\Ypar^{**}]$.

\subsection{Restrictions on the consideration set formation mechanism}
In the presence of limited consideration, the above argument does not directly apply, as one needs to account for all possible consideration sets in which bundle $\cI_{1,1}$ is included.
We therefore need to introduce additional notation and some restrictions.

For any $\cK_1,\cK_2\subseteq \cD$, $\cK_1\cap\cK_2=\emptyset$, denote the probability that all elements of $\cK_1$ are included in the consideration set while all elements of $\cK_2$ are excluded from it, by
\begin{align*}
\cOE(\cK_1;\cK_2)\equiv\sum_{\cK: ~\cK_1 \subset \cK, ~ \cK_2 \cap \cK =\emptyset}\Pr(C_i=\cK|t_i=1)=\sum_{\cK: ~\cK_1 \subset \cK, ~ \cK_2 \cap \cK =\emptyset}\mathcal{Q}_{1}(\cK),
\end{align*}
and define $\cOY(\cK_1;\cK_2)$ similarly, where $\mathcal{Q}_{t}(\cK)$, $t=0,1$, was introduced in Assumption \ref{ass:consideration}.

Denote by $\cB(\cI_{\ell,q},\covx;\zeta)$ the collection of bundles that, at a given value of $\zeta$, strictly dominate bundle $\cI_{\ell,q}$, with $\zeta=\Epar_i$ for agents of type $t_i=1$, and $\zeta=\Ypar_i$ for $t_i=0$:
		\begin{align*}
			\cB(\cI_{\ell,q},\covx;\zeta)&\equiv\{\cI_{k,r} ~s.t.~ CE_\zeta(\cI_{k,r},\covx)> CE_\zeta(\cI_{\ell,q},\covx)\}.
		\end{align*}
Then, for a given value of $\covx$, any bundle $\cI_{\ell,q}\in\cD$ is chosen if and only if it is considered and every bundle that dominates it is not:\footnote{Equivalently, bundle $\cI_{\ell,q}$ is chosen if and only if it is the first best among the ones considered:	
	\begin{align*}
		\Pr(\cI^*=\cI_{\ell,q}|\covx) &=\alpha  \sum\limits_{\cI_{\ell,q}\in\cK} \mathcal{Q}_{1}(\cK)\int \one(CE_\Epar(\cI_{k,r},\covx)\leq CE_\Epar(\cI_{\ell,q},\covx)~ \forall \cI_{k,r}\in \cK| \covx;\Epar)dF\\
		&+(1-\alpha) \sum\limits_{\cI_{\ell,q}\in\cK} \mathcal{Q}_{0}(\cK) \int \one(CE_\Ypar(\cI_{k,r},\covx)\leq CE_\Ypar(\cI_{\ell,q},\covx)~ \forall \cI_{k,r}\in \cK| \covx;\Ypar)dG.
	\end{align*}	
}
	\begin{align}
		\Pr(\cI^*=\cI_{\ell,q}|\covx)&= \alpha \int \cOE(\cI_{\ell,q};\cB(\cI_{\ell,q},\covx;\Epar))dF+ (1-\alpha)\int \cOY(\cI_{\ell,q};\cB(\cI_{\ell,q},\covx;\Ypar))dG.\label{eq:choice_prob}
	\end{align}
Eq.~\eqref{eq:choice_prob} with $(\ell,q)=(1,1)$ shows that $\cI_{1,1}$ is chosen when it is the bundle in $C_i$ with the highest certainty equivalent, i.e., no bundle that yields a higher certainty equivalent (those in $\cB(\cI_{1,1},\covx;\cdot)$) is considered.
Hence, an agent choosing $\cI_{1,1}$ switches to or from a different bundle $\cI_{k,r}$ if and only if (i) they are indifferent between $\cI_{1,1}$ and $\cI_{k,r}$; and (ii) they do not consider any bundle in $\cD$ that dominates $\cI_{1,1}$ and $\cI_{k,r}$. 
As the indifference cutoffs involving bundle $\cI_{1,1}$ are unique, differentiating Eq.~\eqref{eq:choice_prob} we have
\begin{align}
	\frac{\partial\Pr(\cI^*=\cI_{1,1}|\covx)}{\partial \covx^{\cl}} 
	&= \alpha \sum\limits_{(k,r)\neq(1,1)} \cOE(\{\cI_{1,1},\cI_{k,r}\};\cB(\cI_{1,1},\covx;\EEcut^{1,1}_{k,r}))f(\EEcut^{1,1}_{k,r})\frac{\partial \EEcut^{1,1}_{k,r}}{\partial \covx^{\cl}}\notag\\
	&+(1-\alpha) \sum\limits_{(k,r)\neq(1,1)} \cOY(\{\cI_{1,1},\cI_{k,r}\};\cB(\cI_{1,1},\covx;\YYcut^{1,1}_{k,r}))g(\YYcut^{1,1}_{k,r})\frac{\partial \YYcut^{1,1}_{k,r}}{\partial \covx^{\cl}}.\label{Derivate_ij}
\end{align}
The summation in Eq.~\eqref{Derivate_ij} collects all relevant consideration sets across preference types and indifference points (cutoffs), weighted by the density function at these indifference points and taking into account how the change in $\covx^\cl$ affects the indifference points themselves.\footnote{For $\frac{\partial\Pr(\cI^*=\cI_{1,1}|\covx)}{\partial \covx^{\cm}}$, the right-hand-side of Eq.~\eqref{Derivate_ij} remains as is, with $\partial \covx^\cm$ replacing $\partial \covx^{\cl}$.}

We impose the following restrictions on the consideration set formation mechanism:
\begin{assumption}[Minimally informative consideration]\label{ass:MIC}
One of the following holds: 
\begin{enumerate}[label=(\Roman*)]
\item \label{MIC1}
$\cOE(\{\cI_{1,1},\cI_{2,2},\cI_{2,1}\};\emptyset)=\cOE(\{\cI_{1,1},\cI_{2,2},\cI_{1,2}\};\emptyset)>0.$
\item \label{MIC2}
$\cOE(\{\cI_{1,1},\cI_{2,2},\cI_{2,1}\};\emptyset)-\cOE(\{\cI_{1,1},\cI_{2,2},\cI_{1,2}\};\emptyset)\neq 0$, and
\item[]$\cOE(\{\cI_{1,1},\cI_{2,1}\};\emptyset)=\cOE(\{\cI_{1,1},\cI_{2,1}\}; \{\cI_{2,2},\cI_{1,2}\})$.\footnote{Alternatively,  $\cOE(\{\cI_{1,1},\cI_{1,2}\};\emptyset)=\cOE(\{\cI_{1,1},\cI_{1,2}\}; \{\cI_{2,2},\cI_{2,1}\})$ can replace the last condition in Assumption \ref{ass:MIC}-\ref{MIC2}. 
In our application this alternative restriction is satisfied because bundle $\cI_{1,2}$ (which is the deductible bundle $\{\$1000,\$500\}$) is chosen with probability zero, and hence both probabilities are zero.}
\end{enumerate}
One of these two restrictions also holds with $\cOY$ replacing $\cOE$.
\end{assumption}
Assumption \ref{ass:MIC}-\ref{MIC1} requires symmetry in the probability with which the triplets $(\cI_{1,1},\cI_{2,2},\cI_{1,2})$ and $(\cI_{1,1},\cI_{2,2},\cI_{2,1})$ are included in the consideration set, and that each probability is strictly positive, so that information can be extracted through the differentiation in Eq.~\eqref{Derivate_ij}.
Assumption \ref{ass:MIC}-\ref{MIC2} requires that if such symmetry is absent,  then alternatives $\cI_{1,1}$ and $\cI_{2,1}$ can only be considered together when neither $\cI_{1,2}$ nor $\cI_{2,2}$ are considered (a trivial case that would guarantee this condition is that $\cI_{2,1}$ is never considered when $\cI_{1,1}$ is).
The conditions in Assumption \ref{ass:MIC} are sufficient (together with the other assumptions listed above) for our identification results.
However, they can be replaced by technical yet verifiable assumptions on the behavior of the cutoffs involving comparisons of alternatives $\cI_{1,1},\cI_{2,1},\cI_{1,2},\cI_{2,2}$.\footnote{These conditions are available from the authors upon request, and require that $\partial\EEcut^{1,1}_{1,2}(\covx)/\partial \covx^{\cm}$ does not equal a specific linear function of $\partial\EEcut^{1,1}_{2,1}(\covx)/\partial \covx^{\cl}$.}

\subsection{Point identification results}\label{subsec:identif_results}
We next state our main identification results, whose proofs are in the Appendix.
\begin{theorem}\label{T1alt}
	Let Assumptions \ref{ass:coeff_restrict}, \ref{ass:stability}, \ref{ass:narrow}, \ref{ass:SCP}, \ref{ass:consideration}, \ref{ass:DGP}, \ref{ass:not_dominated}, \ref{ass:distinct_contexts}, \ref{ass:var_x}, \ref{ass:MIC} hold.  Then
	\begin{enumerate}
		\item $f(\cdot)$ is identified up to scale on any interval $[\Epar^*,\Epar^{**}]$ satisfying Assumption \ref{ass:var_x}.
		\item $g(\cdot)$ is identified up to scale on any interval $[\Ypar^*,\Ypar^{**}]$ satisfying Assumption \ref{ass:var_x}.
		\item If $[\Epar^*,\Epar^{**}]=[0,\bar\Epar]$ and $[\Ypar^*,\Ypar^{**}]=[0,\bar \Ypar]$, then $f(\cdot)$ and $g(\cdot)$ are identified.
	\end{enumerate}
\end{theorem}
Theorem \ref{T1alt} shows that under limited consideration, despite the lack of independent variation in observed covariates across alternatives (within a single context), it is nonetheless possible to identify the distribution of the random coefficient for each preference type without relying on identification at infinity arguments.\footnote{If one had variation in $\covx^j$ across alternatives and unbounded support, letting the observed covariate (say, price) for a given alternative go to infinity would be akin to assuming that one observes agents repeated choices in context $j$ while facing feasible sets that include/exclude each single alternative.}
While to pin down the entire distribution of preferences large support is required, our approach identifies (up to scale) the density function of each random coefficient conditional on a given interval.
Let $\bar{V}$ (respectively, $\bar{W}$) denote the union of all intervals $[\Epar^*,\Epar^{**}]$ (respectively, $[\Ypar^*,\Ypar^{**}]$) satisfying Assumption \ref{ass:var_x}.
If $\bar{V}$ is a proper subset of $[0,\bar\Epar]$ (respectively, $\bar{W}$ is a proper subset of $[0,\bar\Ypar]$), partial identification of the entire distribution of preferences is still possible, by collecting the probability distribution functions that have density equal to $f(\Epar)$ for all $\Epar\in\bar{V}$ (respectively, $g(\Ypar)$ for all $\Ypar\in\bar{W}$).
For a general treatment of partial identification of preferences in discrete choice models with limited consideration, see \cite{BCMT21}.

One can point identify the shares of preference types under a mild additional restriction, where the probability of including one specific pair of bundles in the consideration set and excluding another specific bundle (or pair of bundles) is independent of preference type.
\begin{corollary}\label{cor:T1}
	$\alpha$ is identified if all Assumptions of Theorem \ref{T1alt} hold, and either:
	\begin{enumerate}
	\item[(i)] Assumption \ref{ass:MIC}-\ref{MIC1} holds for both agents with preference types $t_i=1$ and $t_i=0$, and
	\[\cOE(\{\cI_{1,1},\cI_{2,1}\};\emptyset)-\cOE(\{\cI_{1,1},\cI_{2,1}\}; \{\cI_{2,2},\cI_{1,2}\})=\cOY(\{\cI_{1,1},\cI_{2,1}\};\emptyset)-\cOY(\{\cI_{1,1},\cI_{2,1}\}; \{\cI_{2,2},\cI_{1,2}\}).\]
	\item[(ii)] Assumption \ref{ass:MIC}-\ref{MIC2} holds for both agents with preference types $t_i=1$ and $t_i=0$, and
	\[\cOE(\{\cI_{1,1},\cI_{2,2}\}; \cI_{2,1})-\cOE(\{\cI_{1,1},\cI_{2,2}\}; \cI_{1,2})=
      \cOY(\{\cI_{1,1},\cI_{2,2}\}; \cI_{2,1})-\cOY(\{\cI_{1,1},\cI_{2,2}\}; \cI_{1,2}).\]
\end{enumerate}	
\end{corollary}

Given the distributions of the random coefficients, $F(\cdot)$ and $G(\cdot)$, the system of equations defined in Eq.~\eqref{eq:choice_prob} ($L\times M$ equations for a given $\covx$) is linear in the consideration probabilities across the two types, weighted by their respective shares $\alpha$ and $1-\alpha$. 
This in turn implies that we have a continuum of $L\times M$ linear equations to pin down $2^{L\times M+1}$ parameters. 
In general, with sufficient variation in $\covx$, these parameters are over-identified, subject to standard non-redundancy assumptions.\footnote{For example, if for type $t_i=1$ alternative $\cI_{\ell,k}$ dominates alternative $\cI_{q,r}$, $\mathcal{Q}_1(\{\cI_{\ell,k},\cI_{q,r}\})$ cannot be separately identified from $\mathcal{Q}_1(\{\cI_{\ell,k}\})$.}
However, depending on the specific models of preferences assumed, and on the richness of variation in the data observed, it may not be possible to identify some parts of the distribution of consideration sets.
Nevertheless, for a specific model, given the data, one can test whether a full rank system of equations results across observed values of $\covx$ \citep[see., e.g.,][]{chen:fang19}.

More broadly, our limited consideration model has several testable implications. 
We highlight two: one specific to our broad consideration case, the other more general.
First, suppose Assumption \ref{ass:narrow} holds. 
Then under full or narrow consideration, the marginal distribution of choices in context $\cl$ is invariant to changes in $\covx^\cm$ and vice versa. 
Under broad consideration this is not the case, as can be seen through a simple example where $\mathrm{card}(\cD^j)=2$ for both $j=\cl$ and $j=\cm$, and a positive share of agents consider only the two bundles $\{\cI_{1,1}, \cI_{2,2}\}$.
Hence, one can test for violations of a narrow consideration model by checking whether the marginal distribution of choices in context $\cl$ (respectively, $\cm$) responds to changes in $\covx^\cm$ (respectively, $\covx^\cl$).
A second testable implication of the model is obtained as follows.
Recall that our identification argument focuses on the cheapest bundle, $\cI_{1,1}$, and is built by looking at how its share responds to changes in $\covx^\cl$ and $\covx^\cm$. 
An identical argument can be constructed by focusing on the most expensive bundle, $\cI_{M^\cl,M^\cm}$. 
Hence, the density functions $f(\Epar)$ and $g(\Ypar)$ can be recovered through two different channels. 
If they do not coincide, this implies that at least one modeling assumption is violated.

We conclude by comparing the amount of variation in $\covx=(\covx^\cl,\covx^\cm)$ that we require for our point identification results, with that required in the closely related prior work of \cite{BaMoTh21} to obtain semi-nonparametric point identification of a model with a single preference type.
\citeauthor{BaMoTh21}'s results are derived for an environment where agents are observed making choices only in a single context and with a single source of independent data variation, say context $\cl$ with variation in $\covx^\cl$.
The covariate $\covx^\cl$ is assumed to vary independently across agents; however, for a given agent there is no requirement of independent variation in $\covx^\cl$ across alternatives in $\cD^\cl$ (similarly to this paper). 
Due to the less rich choice environment observed, to recover the conditional distribution of preferences, \cite{BaMoTh21} impose stronger restrictions than we do here on the consideration set formation mechanism.\footnote{For example, \cite{BaMoTh21} require that whenever $\ell^{1\cl}$ is considered, $\ell^{2\cl}$ is also considered.
They do so because there is not a one-to-one mapping between $\partial\Pr(\cI^*=\cI_{1}|\covx)/\partial \covx^{\cl}$ and the (up-to-scale) density function evaluated at a single point. 
Rather, $\partial\Pr(\cI^*=\cI_{1}|\covx)/\partial \covx^{\cl}$ maps into a linear combination of the density function evaluated at cutoffs $\EEcut^{1}_{k}(\covx^{\cl}), k>1$.
In contrast, here by properly utilizing variation in $\covx^{\cm}$ we are able to create such a mapping even though there can be multiple preference types.}
	 	  
\section{Model \& Data on Choices in Automobile Insurance}
\label{sec:empirical}

\subsection{Empirical model}
\label{Model}

As introduced in Section \ref{subsection:lotteries}, we model agents' choices in two contexts of insurance coverage, where each coverage provides full insurance against covered losses in excess of a deductible chosen by the agent.
In our data, the decision maker is a household; hence, we refer to agents as households.
As a reminder, $\mu_i^j$ denotes the probability of household $i$ experiencing a claim in context $j$; for each coverage $j\in \{\cl,\cm\}$, household $i$ faces a menu of premium-deductible pairs, $\mathcal{M}_i^j\equiv\{(\dor^{\ell j},\covx_i^{\ell j}):\ell\in\cD^j\}$, where $\covx_i^{\ell j}$ is the household-specific premium associated with deductible $\dor^{\ell j}$ and $\cD^j$ is the set of deductible options offered in context $j$.
As discussed in Section \ref{subsection:lotteries}, for each context $j\in \{\cl,\cm\}$ the ratio of the price of deductible $\dor^{\ell j}$ to the price of deductible $\dor^{k j}$ is constant across households for all $\dor^{\ell j},\dor^{k j}\in\cD^j$.

We make assumptions, that are widespread in the literature on property insurance, related to filing claims and their probabilities:
\begin{assumption}[Restrictions Related to Claim Probabilities]
\label{ass:claim}
{\color{white}line}
\begin{enumerate}[label=(\Roman*)]
\item \label{ass:one_claim} Households disregard the possibility of experiencing more than one claim during the policy period.
\item \label{ass:claim_props} Any claim exceeds the highest available deductible; payment of the deductible is the only cost associated with a claim; the household's deductible choice does not influence its claim probability.
\end{enumerate}
\end{assumption}

We assume that the two types of preferences described in Section \ref{subsection:types} result from either Expected Utility Theory (EU) or Yaari's \citeyearpar{Yaari1987} Dual Theory (DT).
Within EU, a single-context lottery is evaluated through
\begin{align}
U_i(\cL(\dor^{\ell j},\covx_i^{\ell j},\mu_i^j))\equiv(1-\mu_i^j)u_{i}(w_{i}-\covx_i^{\ell j})+\mu_i^ju_{i}(w_{i}-\covx_i^{\ell j}-\dor^{\ell j}),\label{eq:EUT}
\end{align}
where $w_{i}$ is the household's wealth and $u_{i}(\cdot)$ is its Bernoulli utility function, which under Assumption \ref{ass:stability} is the same for each context. 
In the EU model, utility is linear in the probabilities and aversion to risk is driven by the shape of the utility function $u_{i}(\cdot)$.

Yaari's \citeyearpar{Yaari1987} DT model aims at decoupling the decision maker's attitude towards risk from her attitude towards wealth. 
Within DT, a single-context lottery is evaluated through
\begin{align}
U_i(\cL(\dor^{\ell j},\covx_i^{\ell j},\mu_i^j))\equiv(1-\Omega_{i}(\mu_i^j))(w_{i}-\covx_i^{\ell j})+\Omega_{i}(\mu_i^j)(w_{i}-\covx_i^{\ell j}-\dor^{\ell j}),\label{eq:Yaari}
\end{align}
where $\Omega_{i}(\cdot)$ is the household's probability distortion function, which under Assumption \ref{ass:stability} is the same for each context. 
In the DT model, utility is linear in the outcomes and aversion to risk is driven by the shape of the probability distortion function $\Omega_{i}(\cdot)$.\footnote{Probability distortions are featured also in, e.g., prospect theory \citep{Kahneman1979,Tversky1992}, rank-dependent expected utility theory \citep{Quiggin1982}, \citet{Gul1991} disappointment aversion theory, and \citet{Koszegi2006,Koszegi2007} reference-dependent utility theory.} 
We remark that in our setting (as well as in many others where subjective beliefs data are not collected and the analysis relies on an often implicit rational expectations assumption), the DT model is indistinguishable from one in which agents' subjective loss probabilities systematically deviate through the $\Omega_{i}(\cdot)$ function from the objective ones.

To strike a balance between model generality and its empirical tractability, we impose shape restrictions on $u_{i}(\cdot)$ and $\Omega_{i}(\cdot)$, respectively. 
We assume $u_{i}(\cdot)$ exhibits constant absolute risk aversion (CARA):
\begin{assumption}[CARA]
\label{ass:cara}
$u_i(y)=\frac{1-\exp(-\Epar_i y)}{\Epar_i}$ for $\Epar_i\neq 0$ and $u_i(y)=y$ for $\Epar_i=0$. 
\end{assumption}
Assuming CARA has two key virtues. 
First, $u_{i}(\cdot)$ is fully characterized by a single parameter: the Arrow-Pratt coefficient of absolute risk aversion, $\Epar_{i}\equiv-u_{i}^{\prime\prime}(w_{i})/u_{i}^{\prime}(w_{i})$. 
Second, $\Epar_{i}$ is a constant function of $w_{i}$, and hence we need not observe wealth to estimate $u_{i}(\cdot)$.

To keep the EU model and the DT model on ``equal footing," we need $\Omega_{i}(\cdot)$ to be as parsimonious as $u_{i}(\cdot)$. 
This suggests a single-parameter specification. 
The literature contains many examples, and we run our analysis with the following one due to \citet{Prelec1998}:
\begin{assumption}[\citeauthor{Prelec1998}'s $\Omega(\cdot)$ function]
\label{ass:Omega}
$\Omega_i(\mu)=\exp(-(-\ln\mu)^{\Ypar_i})$, $\Ypar_i>0$.
\end{assumption}
We also carry out our analysis using other utility functions for the EU type (one proposed by \citet{Cohen2007} and one by \citet{Barseghyan2013}) and other probability distortion functions for the DT type (one put forward by \citet{Tversky1992} and one  by \citet{BMT16}). 
The results confirm the main takeaways reported here, and are available from the authors upon request.\footnote{Vuong tests comparing the various models confirm the good fit of our preferred specification.}

The EU and DT models are true alternative theories of decision making under risk.\footnote{Except when both degenerate into net present value calculations with $\Epar_i=0$ and $\Ypar_i=1$.} 
Neither model is a special case of the other. 
DT preferences depart from EU preferences in two key ways. 
First, risk averse behavior is driven by distortions of probabilities for households with DT preferences, but by nonlinear evaluation of wealth for households with EU preferences. 
Second, narrow bracketing has behavioral implications for households with DT preferences, but not for households with EU preferences. 
In our framework, where the lotteries are independent across the brackets,\footnote{Independence results from the assumption that claims follow a Poisson distribution, which is imposed in estimating the probability of a claim \citep[see][]{Barseghyan2013,BTX18}.} the choices of a household with EU preferences and CARA utility are independent of the scope of bracketing \citep[e.g.,][]{Rabin2009}. 
The well-known reason is the absence of wealth effects with CARA utility. 
In contrast, the choices of a household with DT preferences are not independent of the scope of bracketing, because of the rank-dependent nature of how probability distortions are applied.

Within context $j$, the resulting utility function is
\begin{equation}
U_i(\cL(\dor^{\ell j},\covx_i^{\ell j},\mu_i^j)) = \left\{ \begin{tabular}{ll}
$(1-\mu_i^j) u_i(w_i-\covx_i^{\ell j})+\mu_i^j u_i(w_i-\covx_i^{\ell j}-\dor^{\ell j})$ &if $t_i = 1$~~(EU), \\ 
$( 1-\Omega_i(\mu_i^j)) (w_i-\covx_i^{\ell j})+\Omega_i(\mu_i^j)(w_i-\covx_i^{\ell j}-\dor^{\ell j})$ & if $t_i=0$~~(DT).\\ 
\end{tabular}  \right.\label{eq:U_final}
\end{equation}

While we obtain conditions for nonparametric point identification of $F(\cdot)$ and $G(\cdot)$, for tractability we estimate a fully parametric model via Maximum Likelihood.\footnote{Inspection of Eqs.~\eqref{eq:S_Epar}-\eqref{eq:S_MIC1}-\eqref{eq:S_MIC2} in the Appendix shows that under Assumption \ref{ass:het_restr}, $f(\cdot)$ and $g(\cdot)$ are identified, provided the intervals $[\Epar^*,\Epar^{**}]$ and $[\Ypar^*,\Ypar^{**}]$ in Assumption \ref{ass:var_x} are not singletons.}
\begin{assumption}[Heterogeneity Restrictions]
\label{ass:het_restr}
{\color{white}a}
\begin{enumerate}[label=(\Roman*),topsep=1ex,itemsep=0pt]
\item \label{ass:beta_r} Conditional on $t_i=1$, $\Epar_i$ follows a Beta distribution on $[0,0.025]$ with parameter vector $(\gamma_{\Epar 1},\gamma_{\Epar 2})$ and is independent of $[(\mu_i^j,\covx_i^j),j=\cl,\cm]$. 
\item \label{ass:beta_omega} Conditional on $t_i=0$, $\Ypar_i$ follows a Beta distribution on $[0,1]$ with parameter vector $(\gamma_{\Ypar 1},\gamma_{\Ypar 2})$ and is independent of $[(\mu_i^j,\covx_i^j),j=\cl,\cm]$. 
\end{enumerate}
\end{assumption}
Assumption \ref{ass:het_restr} specifies that the distributions of $\Epar$ and $\Ypar$ are Beta distributions. The main attraction of the Beta distribution is its flexibility \citep[e.g.,][]{Ghosal2001}. Its bounded support is a plus given our setting. A lower bound of zero rules out risk-loving preferences and seems appropriate for insurance markets that exist primarily because of risk aversion. 
Imposing an upper bound enables us to rule out absurd levels of risk aversion.
The choice of 0.025 for CARA 
is conservative both as a theoretical matter and in light of prior empirical estimates in similar settings \citep[e.g.,][]{Cohen2007,Sydnor2010,Barseghyan2011,Barseghyan2013,BMT16}. 
Similarly, for the probability distortion function, the upper bound of 1 insures over-weighting of probabilities; the lower bound of 0 insures that it is a well-behaved function. 
None of these constraints is binding in our analysis. 

We close the empirical model by restricting how $C_i\subseteq\cD=\cD^\cl\times\cD^\cm$ is drawn:
\begin{assumption}[(Broad) Alternative-Specific Consideration]
\label{ass:ARC}
Household $i$ draws a \emph{consideration set} $C_i\subseteq\cD$ s.t.
\begin{align*}
\Pr(C_i = G) = \prod_{\cI\in G}\arc_\cI \prod_{\tilde{\cI}\notin G} (1-\arc_{\tilde{\cI}}),~~\forall G\subseteq\cD,
\end{align*}
where $\arc_\cI\equiv\Pr(\cI\in C_i)=\Pr(\cI\in C_i|t_i)\ge0,~\cI\in\cD$, and $\arc_{\cI_{1,1}}=1$.
\end{assumption}
Assumption \ref{ass:ARC} strengthens Assumption \ref{ass:consideration} by requiring consideration to be independent of type (in addition to being independent of households' preferences given type).
This is \emph{not} needed to establish identification, but we think it prudent to impose it in our application because, as further discussed below, $|\cD|=30$ and allowing for type-dependent consideration would add 60 rather than 30 consideration parameters to the model. 
Assumption \ref{ass:ARC} also adapts the Alternative-specific Random Consideration (ARC) model first proposed by \citet{Manski1977} and later axiomatized by \citet{man:mar14}, to hold over bundles of insurance deductibles across contexts.
Each bundle $\cI\in\cD$ appears in the consideration set with probability $\arc_\cI$ independently of other bundles. 
To avoid empty consideration sets, following \citet{Manski1977}, we assume that one bundle is always considered, and further impose that the always-considered bundle is the cheapest one.\footnote{Alternatively, we could assume that if the realized consideration set is empty, agents choose one of the alternatives in $\cD$ uniformly at random. Our estimation results are robust to this modeling assumption.}
Once the consideration set is drawn, the household chooses the best alternative according
to its preferences as in Eq.~\eqref{eq:max_CE}. 

\subsection{Data Description}

\label{sec:data}
We obtained the data from a large U.S. property and casualty insurance company. 
The company offers several lines of insurance, including auto. 
As explained in Section \ref{subsection:lotteries}, we focus on deductible choices in auto collision and auto comprehensive.
Our analysis uses a sample of 7,736 households who purchased their auto and home policies for the first time between 2003 and 2007 and within six months of each other (this is the same sample used by \citet{BaMoTh21}).\footnote{As explained in \citet{BaMoTh21}, the dataset is an updated version of the one used in \citet{Barseghyan2013}. It contains information for an additional year of data and puts stricter restrictions on the timing of purchases across different lines. These restrictions are meant to minimize potential biases stemming from non-active choices, such as policy renewals, and temporal changes in socioeconomic conditions.}
We observe households' deductible choices in auto collision and auto comprehensive, and the premiums they paid for these coverages. 
We also observe the household-coverage specific menus of deductible-premium combinations---i.e., the pricing menus---that were available to the households when they made their deductible choices.

We refer to Section \ref{subsection:lotteries} for a discussion of how households' pricing menus are determined by the company in each context.
As explained there, in each context the premium $\covx_i^{\ell j}$ associated to deductible $\dor^\ell,\ell\in\cD^j$, is a household-invariant affine function of a household-specific base price $\covx_i^j$, and the company determines this base price applying a coverage-specific rating function to household $i$'s coverage-relevant characteristics. 
Naturally, the base prices $\covx_i^\cl$ and $\covx_i^\cm$ may exhibit substantial correlation due to common factors entering the rating function (this correlation equals 0.74 in our data), highlighting the importance of our weak requirement on variation in $\covx$ stated in Assumption \ref{ass:var_x} -- which in particular can hold when $\covx^\cl$ and $\covx^\cm$ are strongly correlated (see Figure \ref{Fig:boldX} and its discussion).

\begin{table}\caption{Collision and Comprehensive Deductible Choices, in \% }\label{Table_jointD}
\begin{tabularx}{1\textwidth}{c *{7}{Y}}
\\ 
\hline
\textbf{}     & \multicolumn{6}{c}{\textbf{Comprehensive}}                                                                                                                             \\
\multicolumn{1}{l}{\textbf{Collision}}   & \multicolumn{1}{c}{\$50} & \multicolumn{1}{c}{\$100} & \multicolumn{1}{c}{\$200} & \multicolumn{1}{c}{\$250} & \multicolumn{1}{c}{\$500} & \multicolumn{1}{c}{\$1,000} \\ \hline
\$100                 & 0.7                      & 0.2                       & 0                         & 0                         & 0                         & 0                           \\
\$200                 & 1.8                      & 1.1                       & 10                      & 0                         & 0.1                         & 0                           \\
\$250                 & 0.9                      & 1.3                       & 4.6                       & 5.4                       & 0                         & 0                           \\
\$500                 & 1.0                      & 1.3                       & 17.8                      & 6.5                       & 41                      & 0                           \\
\$1,000               & 0                        & 0.1                         & 0.4                       & 0.2                       & 1.9                       & 3.7        \\ \hline               
\end{tabularx}
\end{table}

Table \ref{Table_jointD} reports the deductible choices of the households in our sample. 
In each context, the modal choice is \$500. 
Interestingly, virtually no household purchases a comprehensive deductible larger than their collision deductible.
As we discuss in more detail below, this choice pattern cannot be rationalized by standard discrete choice models under the assumption of full consideration, but can easily be explained once one allows for limited consideration.

The top panel of Table \ref{Table_quantiles} shows that base premiums vary dramatically in our sample. 
The ninety-ninth percentile of the \$500 deductible is more than ten times the corresponding first percentile in each line of coverage.
While not reported in the table, here we summarize the pricing menus.
The cost of decreasing the deductible from \$500 to \$250 is on average \$56 in collision and \$31 in comprehensive.
The saving from increasing the deductible from \$500 to \$1,000 is on average \$42 in collision and \$23 in comprehensive.

The claim probabilities $\mu_i^j$ stem from \citet{BTX18}, who estimated them using coverage-by-coverage Poisson-Gamma Bayesian credibility models applied to a large auxiliary panel of more than one million observations. 
We treat estimated claim probabilities as if they were observed data. 
Predicted claim probabilities (summarized in the bottom panel of Table \ref{Table_quantiles}) exhibit substantial variation: the ninety-ninth percentile claim probability in collision (comprehensive) is 4.3 (12) times higher than the corresponding first percentile. 
Finally, the correlation between claim probabilities and premiums for the \$500 deductible is 0.38 for collision and 0.15 for comprehensive. 
Hence, there is independent variation in both (although our identification results only require independent variation in premiums).

\begin{table}\centering
\caption{Descriptive statistics for premiums of \$500 deductible and claim probabilities}
\label{Table_quantiles}
	\centering
\resizebox{\textwidth}{!}{\begin{tabular}{l l c c c c c c c c c }
		\hline \hline \\[-1.5ex]	
	&	      & \textbf{Mean} & \textbf{Std.} & \multicolumn{7}{c}{\textbf{Qunatiles}}        \\                                                                                                                    	&      &                             & \textbf{Dev.} & & & & & & &\\
	&	      &  & & 0.01   & 0.05  & 0.25 & 0.50 & 0.75 & 0.95 & 0.99\\\\[-1.5ex]
		\hline \\[-1.5ex]
	Premiums	 & & &&&&&&&&\\
	&	Collision   & 187 & 104  &  53   & 74   & 117 & 162 & 227 & 383 & 565 \\\\[-1.5ex]
	&	Comprehensive  & 117 & 86 & 29   & 41   & 69  & 99  & 141 & 242 & 427\\\\[-1.5ex]
		\hline \\[-1.5ex]
	Claim probs	 &&&&&&&&&&\\
	&	Collision   & 0.081 &  0.026 & 0.036 & 0.045 & 0.062 & 0.077 &  0.096 & 0.128 & 0.156 \\\\[-1.5ex]
	&	Comprehensive & 0.023& 0.012 & 0.005 & 0.008 & 0.014 & 0.021 &  0.030 & 0.045 & 0.062 \\\\[-1.5ex]
		\hline \hline\\
	\end{tabular}}
\end{table}

\subsection{Evidence in support of unobserved heterogeneity in $C_i$}\label{subsection:why_unobs_C}
As discussed in, e.g., \citet{BMT16,BCMT21,BaMoTh21}, standard models of risk preferences fail to rationalize some salient data patterns. 
First, in our data the pricing rule in collision coverage is such that (virtually) no household, regardless of their preference type and random coefficient, should choose the $\$200$ deductible under full consideration. 
The reason is that for agents with lower risk aversion (probability distortions) it is dominated by the $\$250$ deductible, and for agents with higher risk aversion (probability distortions) it is dominated by the $\$100$ deductible.\footnote{An analogous fact can be established even if an i.i.d., type-specific, noise term were added to the utility function in Eq. \eqref{eq:U_final} at the coverage level or, more broadly, for any model that abides a notion of generalized dominance formally defined in \citet{BaMoTh21}. } 
A limited consideration model, even in the case where the consideration set forms narrowly (i.e., with $C_i^\cl$ drawn independently from $C_i^\cm$ and $C_i=C_i^\cl\times C_i^\cm$) has no problems explaining such a pattern, because it allows for the $\$200$ deductible to be considered without either $\$100$ or $\$250$. 
Under Assumption \ref{ass:consideration} (consideration sets drawn at the bundle level), that is not necessary, because utility comparisons are at the bundle level.

Second, the joint probability mass function of choices across contexts (see Table \ref{Table_jointD}) exhibits a striking pattern where virtually none of the 7,736 households purchase a deductible in comprehensive that exceeds the deductible they purchase in collision. 
Unless prices (and claim probabilities) exhibit strong negative correlation, a feature that does not occur in our data, standard models (e.g., a Mixed Logit with full consideration) under the assumption of context invariant preferences will struggle to replicate this pattern. 

A final note pertains to modeling limited consideration as operating at the bundle level, rather than independently across contexts.
A model where limited consideration operates independently across contexts may be successful in matching the marginal distribution of choices within each context, but not the joint \citep[see the working paper][Section 7.3.3]{BaMoTh19}. 
The limited consideration model studied in this paper, by operating on the bundles, does have the capacity to match the joint distribution of choices. 
By doing so, it also resolves the preference stability debate discussed in, e.g., \cite{Barseghyan2011,Einav2012,BMT16}.
This debate is centered around the fact that while households’ risk aversion relative to their peers is correlated across lines of coverage, implying that households preferences have a stable component, analyses based on revealed preference reject the standard models: under full consideration, for the vast majority of households one cannot find a level of (household-specific) risk aversion that justifies their choices simultaneously across all contexts. 
Limited consideration allows the model to match the observed joint distribution of choices, and hence their rank correlations.
Under limited consideration, testing for preference stability amounts to asking whether one can find a consideration set and a random coefficient (preference parameter) which jointly rationalize an agent's choice, which is inherently weaker then asking whether one can find preferences that rationalize the agent's choice under full consideration \citep[see, e.g.,][]{BCMT21}.

\section{Estimation Results}\label{sec:results_estimation}

\begin{table}\caption{Estimated Probability of Considering each Deductibles Pair} \label{Table_consideration} 
	\centering
	\begin{tabularx}{\textwidth}{c *{15}{Y}}
		\hline\hline
		\textbf{}     & \multicolumn{12}{c}{\textbf{Comprehensive}}\\
		\multicolumn{1}{l}{\textbf{Collision}}   & \multicolumn{2}{c}{\$50} & \multicolumn{2}{c}{\$100} & \multicolumn{2}{c}{\$200} & \multicolumn{2}{c}{\$250} & \multicolumn{2}{c}{\$500} & \multicolumn{2}{c}{\$1,000} \\ \hline

\$100 & \multicolumn{2}{c}{0.05} &\multicolumn{2}{c}{0.01}&\multicolumn{2}{c}{0.00}& \multicolumn{2}{c}{0.00}& \multicolumn{2}{c}{0.00}& \multicolumn{2}{c}{0.00}  \\ 
& [0.04 & 0.06] & [0.01 & 0.01] & [0.00 & 0.00] &   &   &   &   &   &   \\ 
\hline
\$200 & \multicolumn{2}{c}{0.12} &\multicolumn{2}{c}{0.04}&\multicolumn{2}{c}{0.29}& \multicolumn{2}{c}{0.00} &\multicolumn{2}{c}{0.00}&\multicolumn{2}{c}{0.01} \\ 
& [0.1 & 0.13] & [0.04 & 0.05] &[0.27 & 0.3] & [0.00 & 0.00] & [0.00 & 0.00] & [0.00 & 0.01] \\ 
\hline
\$250&\multicolumn{2}{c}{0.04}&\multicolumn{2}{c}{0.03}&\multicolumn{2}{c}{0.08}&\multicolumn{2}{c}{0.09}&\multicolumn{2}{c}{0.00}&\multicolumn{2}{c}{0.00}\\ 
& [0.03 & 0.05] & [0.02 & 0.03] & [0.07 & 0.08] & [0.09 & 0.10] & [0.00 & 0.00]&   \\ 
\hline
\$500 & \multicolumn{2}{c}{0.13}&\multicolumn{2}{c}{0.06}&\multicolumn{2}{c}{0.46}&\multicolumn{2}{c}{0.18}&\multicolumn{2}{c}{0.83}&\multicolumn{2}{c}{0.00}\\ 
& [0.11 & 0.15] & [0.05 & 0.07] & [0.44 & 0.47] & [0.18 & 0.20] & [0.81 & 0.84] &   &   \\
\hline 
\$1000 & \multicolumn{2}{c}{0.04}&\multicolumn{2}{c}{0.03}&\multicolumn{2}{c}{0.18}&\multicolumn{2}{c}{0.07} &\multicolumn{2}{c}{0.47}&\multicolumn{2}{c}{1.00}\\ 
& [0.02 & 0.07] & [0.01 & 0.05] & [0.14 & 0.22] & [0.05 & 0.10] & [0.43 & 0.52] &   &   \\ 
\hline 
\hline 
\end{tabularx}
     {\raggedright \small Notes: 95\% confidence intervals obtained via subsampling in square brackets.\par}
\end{table}

We begin our discussion of the estimates that we obtain through MLE by focusing on the type of limited consideration that we uncover, and its role in the results one obtains when estimating preferences.
Table \ref{Table_consideration} reports the estimated consideration probabilities for each bundle (these are the $\arc_\cI$ coefficients in Assumption \ref{ass:ARC}), along with 95\% confidence intervals obtained by subsampling.\footnote{We use subsampling because the parameter vector is on the boundary of the parameter space.}
The estimated model is very far from a full consideration one. 
Bundles where the collision deductible is strictly lower than the comprehensive one are almost never considered (the probability that the bundle $(\$200,\$1000$) is considered is 1/100, and all others are zero).\footnote{Given the choice patterns in the data discussed in Section \ref{subsection:why_unobs_C}, this is not surprising, as MLE sets the consideration probability of never-chosen bundles to zero.}
The cheapest bundles, excluding the one where the collision deductible is lower than the comprehensive one, are considered most often (the consideration probabilities for $(\$500,\$500)$ and $(\$1000,\$500)$ are, respectively,  0.83 and 0.47).\footnote{Recall that we assume that $(\$1000,\$1000)$ is considered with probability one.}

The presence of limited consideration alters inference about preference types and about the distribution of the random coefficient within each type in essentially every possible way. 
To illustrate these effects,  we estimate preferences in a pure random coefficients model under three scenarios for the consideration set formation mechanism: limited consideration as in Assumption \ref{ass:ARC} (our proposed model); \emph{triangular consideration}, where for $\cI=[\ell^\cl,q^\cm]$, $\arc_\cI=0$ when $\ell^\cl<q^\cm$ and $\arc_\cI=1$ when $\ell^\cl\ge q^\cm$; and full consideration, where $\arc_\cI=1$ for all $\cI\in\cD$.
In all cases, we estimate a model where households choose their optimal bundle according to Eq.~\eqref{eq:max_CE} with the utility function in Eq.~\eqref{eq:U_final}.\footnote{Under full consideration, the likelihood of observing non-zero shares of never-the-first-best alternatives is zero. Due to this, in estimation we set the consideration probability of each bundle to 0.99 instead of 1.00.}

Figure \ref{Fig:estimatedPrelec} depicts the resulting Prelec distortion function in Assumption \ref{ass:Omega} when $\Ypar_i$ equals the mean, median, 25th and 75th quantile of the distribution $G(\Ypar)$ estimated in the limited consideration model (left panel), in the triangular consideration model (center panel), and in the full consideration model (right panel), each with a mixture of types.
As the figure illustrates, there is substantial variation in the function across these different values of $\Ypar$, and all functions are substantially far from the $45^o$ line, indicating substantial over-weighting of small probabilities.
Of notice is the fact that the over-weighting is larger in the limited consideration model than in the triangular or in the full consideration model. 
\begin{figure}
\centering
\adjincludegraphics[scale=.6,Clip={.05\width} {.2\height} {0.1\width} {.15\height}]{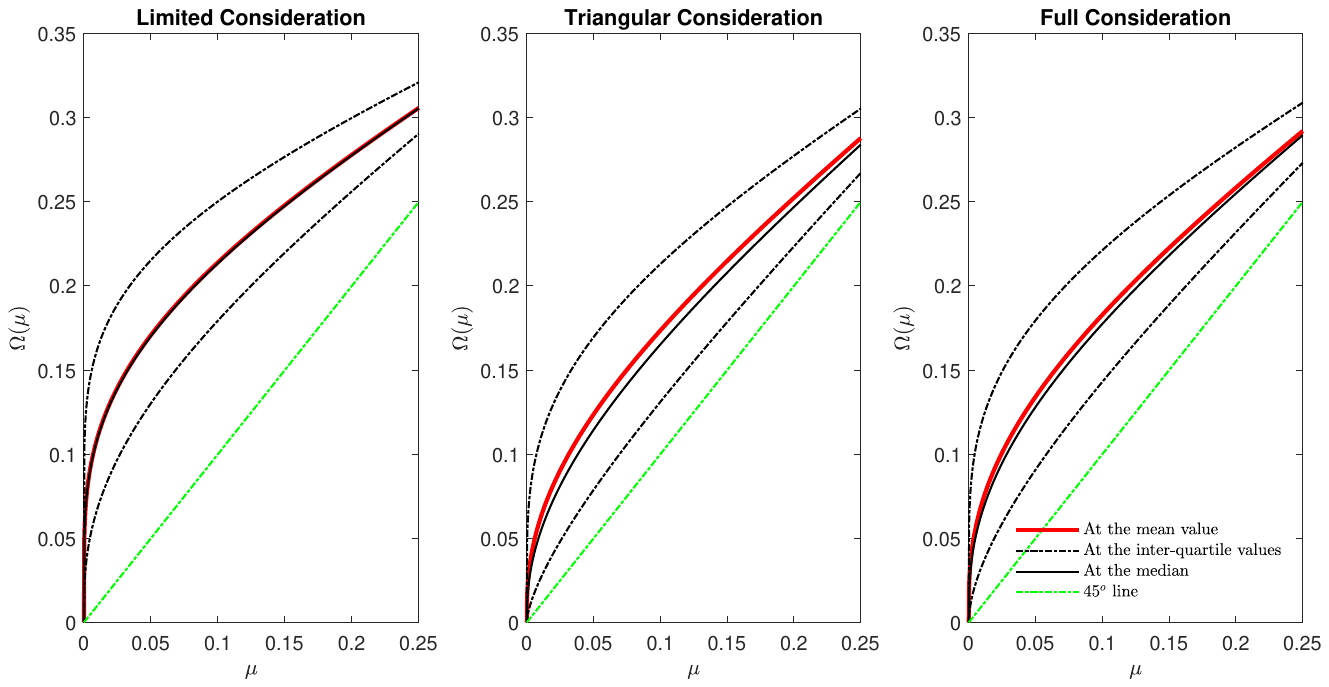}
\caption{\small{The function $\Omega(\mu)$ when $\Ypar$ equals its estimated mean,  25th,  50th,  or 75th quantile, in a model with limited (left), triangular (middle), or full (right) consideration.}}\label{Fig:estimatedPrelec}
\end{figure}

\begin{figure}
\centering
\adjincludegraphics[scale=.6,Clip={.05\width} {.2\height} {0.1\width} {.15\height}]{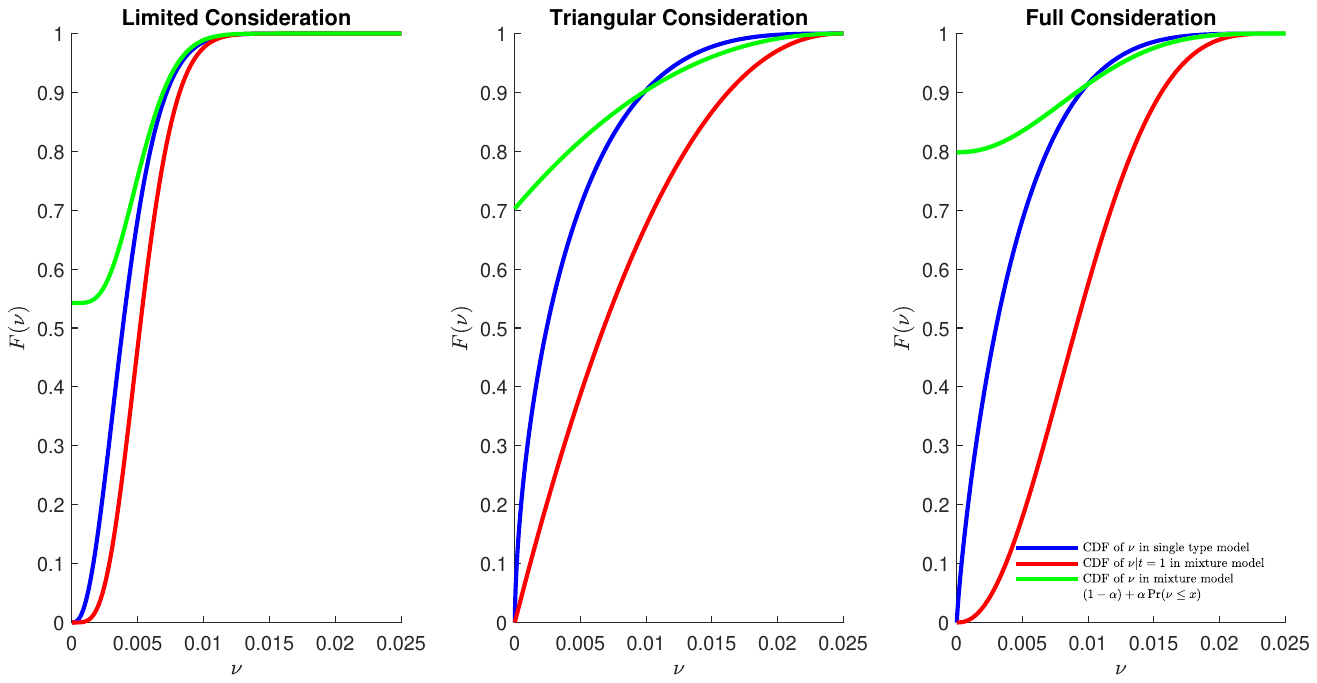}
\caption{\small{Estimated $F(\Epar)$ in a model that assumes limited (left), triangular (middle), or full consideration (right).}}\label{Fig:estimatedCARA}
\end{figure}

Figure \ref{Fig:estimatedCARA} depicts the cumulative distribution function $F(\cdot)$ in our limited consideration model (left panel), in the triangular consideration model (middle panel), and in the full consideration model (right panel).
Each panel depicts $F(\cdot)$ for a model that assumes that all households are of the EU type (blue line), for our model with a mixture of EU and DT types (red line), and, for the mixture model, also the implied cumulative distribution function for the entire population, where the $(1-\alpha)$ share of DT households has $\Epar=0$.
The important feature to notice is that in all panels of Figure \ref{Fig:estimatedCARA}, the risk aversion displayed is much higher for the EU households in the mixture model than in the single-type model, and the discrepancy grows from the limited to the triangular to the full consideration model.

In Table \ref{Table_WTP} we analyze the same interplay between consideration and preferences from a different angle.
We report the estimated excess willingness to pay (WTP) of households in our sample to avoid a lottery where with probability 10\% the household loses \$500 (hence, the total WTP equals $\$50$ plus the values reported in the table).

\begin{table}\caption{\small{Excess willingness to pay to avoid lottery where with probability 10\% agent loses \$500}}\label{Table_WTP}
 {\centering
 \resizebox{\textwidth}{!}{\begin{tabular}{l c c c c c c c c c}
\hline\hline 
        ~ & Mean & Median & $1^{st}$ Quar. & $3^{rd}$ Quar. & ~ & Mean & Median & $1^{st}$ Quar. & $3^{rd}$ Quar. \\
         \hline
         \\
            \textbf{Mixture:} & \multicolumn{4}{c}{\textbf{All Population}}  &\textbf{EU share} & & & &                                                                                                                                                         \\
        Limited Consideration& 82.10 & 79.95 & 61.82 & 100.49  & $\alpha=0.46$ &  &  &  &   \\ 
 Lower Triangular & 73.27 & 72.38 & 42.33 & 101.64 & $\alpha=0.30$ &  &  &  &  \\ 
        Full Consideration&73.49 & 72.42 & 54.31 & {\color{white}1}91.72  & $\alpha=0.20$ &  &  &  &  \\ 
         & \multicolumn{4}{c}{\textbf{EU type}}                                                                                                                                                          & & \multicolumn{4}{c}{\textbf{DT type}}\\
        Limited Consideration& 110.74 & 104.47 & {\color{white}1}69.29 & 146.60 & & 57.95 & 56.73 & 39.70 & 75.11  \\ 
 Lower Triangular & 155.56 & 148.83 & {\color{white}1}50.94 & 255.07 &  & 38.47 & 32.59 & 15.54 & 56.49 \\ 
        Full Consideration &194.83 & 205.93 & 127.16 & 267.12 &  & 42.88 & 38.72 & 21.52 & 60.59 \\ 
        \\
        \hline
          \\
      \textbf{Single Type:} & \multicolumn{4}{c}{\textbf{All Population EU}}                                                                                                                                                          & ~& \multicolumn{4}{c}{\textbf{All Population DT}}\\
        Limited Consideration
         & {\color{white}1}81.74 &  {\color{white}1}69.72 &  {\color{white}1}39.58 & 113.45 & & 74.01 & 75.03 & 51.79 & 97.15
        \\
        Lower Triangular & {\color{white}1}76.19 & {\color{white}1}37.61 & {\color{white}11}8.64 & 121.29 &  & 43.07 & 37.26 & 17.60 & 63.90 \\ 
        Full Consideration &       {\color{white}1}80.87 & {\color{white}1}49.17 & {\color{white}1}15.01 & 127.35 &   & 47.16 & 42.99 & 23.38 & 67.57  \\ 
        \\
        \hline \hline
    \end{tabular}}}

{\flushleft\footnotesize{\emph{Notes}. Top panel: excess WTP in our model for the overall population and within each preference type. 
                Bottom panel: excess WTP for a single-type model, where all agents are either EU or DT.}}
                \end{table}

A first feature to notice is that the estimated share of EU types is much higher when the model allows for limited consideration than in models that assume triangular or full consideration (almost a half versus 30\% and 20\% respectively). 
The implied degree of aversion to risk changes for households of both preference types, but in opposite directions. 
The top left panel of Table \ref{Table_WTP} shows that if one disregards limited consideration, one infers that the risk aversion of EU types is much higher (more than 40\% according to our metric) than under limited consideration, but the aversion to risk of DT types is about one third lower under full consideration (and similarly for triangular consideration).
The cumulative effect of limited consideration in the overall population results in a near 12 percent higher willingness to pay to avoid the simple lottery relative to a model that imposes full consideration.\footnote{These results are sensitive to the choice of the simple lottery to benchmark willingness to pay.  Changing the stakes will induce a non-linear response by the EU types but a linear one by the DT types. Changing the loss probability will induce a non-linear response by the DT types but a linear one by the EU types. }

We conclude by observing that both the full and the triangular consideration model cannot rationalize the choices of a substantial fraction of households in our data and in general deliver a poor fit, as shown in Figure \ref{Fig:model_fit}.
Even adding an Extreme Value Type I error term to the utility function in Eq.~\eqref{eq:U_final} and estimating a Mixed Logit model does not remedy this problem.
Indeed, the Mixed Logits do not fit our data well, while our limited consideration model essentially replicates the observed shares. 

For completeness, in the figure we also display the fit of a limited consideration model where consideration is narrow and choice follows from Eq.~\eqref{eq:choice_narrow_C}. 
While this model fits the data well relative to the Mixed Logit models with full or triangular consideration (compare the third panel to the top two panels in Figure \ref{Fig:model_fit}), it falls short of our benchmark model. 
This is not surprising: by construction, this model is restrictive in how bundles enter the consideration sets.
As a result, it cannot, e.g., set the shares of bundles with $\dor^\cl<\dor^\cm$ to zero, or match certain features of the joint distribution of chosen alternatives in the two contexts, such as rank correlations of choices across the two different coverages.\footnote{The narrow consideration model implies a rank correlation of .42 while in the data and under the broad consideration model this coefficient equals .61 and .62, respectively. In comparison, in the Mixed Logit model with full consideration this correlation is .45, while with lower triangular consideration it is .65.}

\begin{figure}
\centering
\includegraphics[scale=.65]{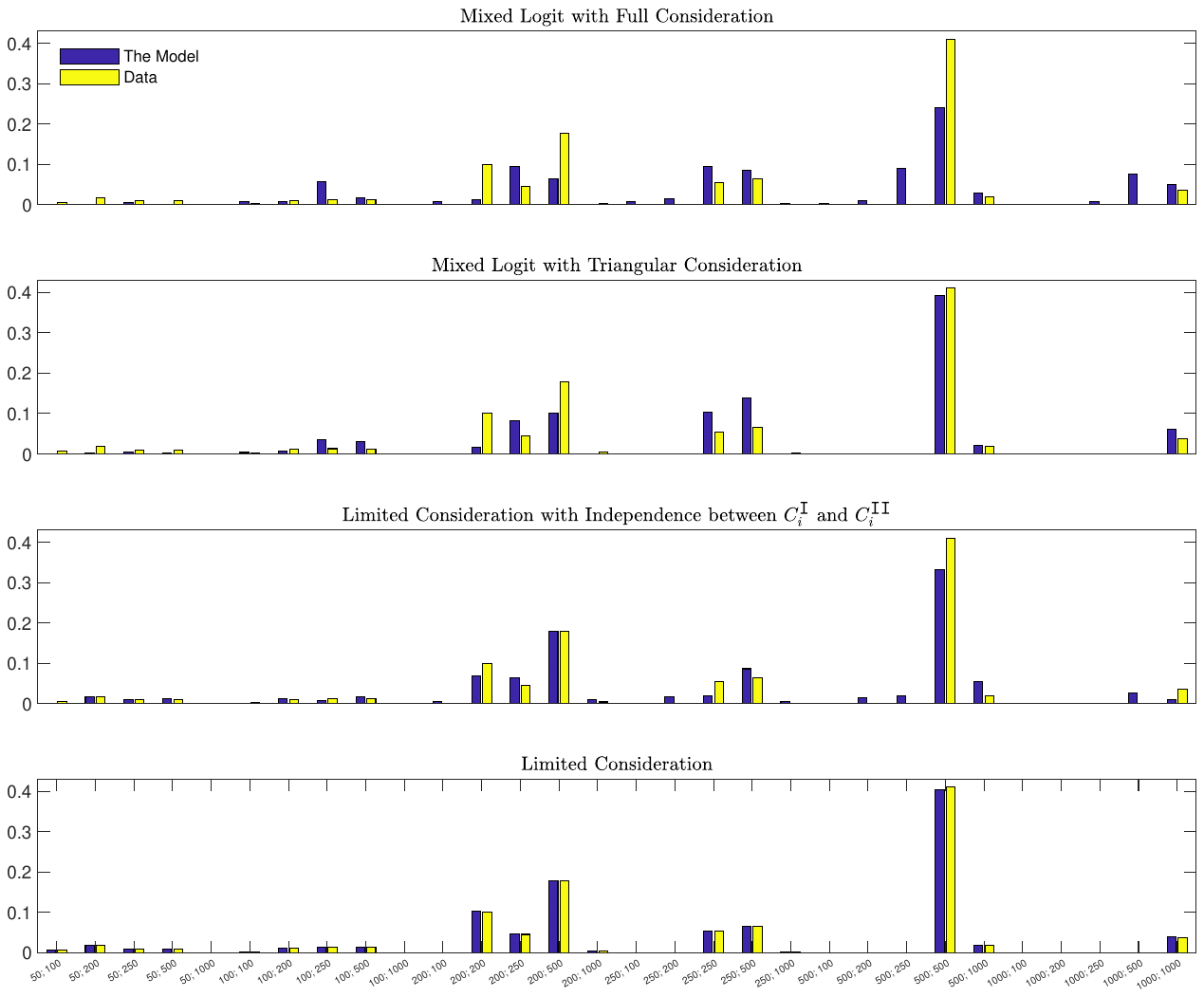}
\caption{\small{Choice probabilities for deductible bundles in a Mixed Logit model with full consideration and with triangular consideration (first and second panel), and in a limited consideration model with narrow and with broad consideration (third and forth panel).}}\label{Fig:model_fit}
\end{figure}

\section{Implications for Welfare Analysis}\label{sec:welfare_results}

In our setting, there are three channels for potential welfare losses. 
First, limited consideration may prevent agents from choosing their first best. 
Second, if the probability distortions are capturing a mismatch between subjective and objective beliefs about loss probabilities,\footnote{See, e.g., the model with imperfect information in \citet[Example 2]{gua:sin23}.} agents may not choose their objective first best, even if they consider it.
Third, non-expected utility maximizing households (the DT type in our model) may be open to \emph{nudging}, whereby modifications of market features that leave the behavior (and welfare) of EU households mostly unchanged may trigger large changes in behavior (and welfare) of DT households.

We therefore conduct two welfare exercises aimed at assessing the impact of each of these channels on the welfare of households purchasing auto deductible insurance.
In the first exercise, we estimate the impact on welfare of all households having full consideration.
To do so, we take the preferences estimated using our limited consideration model, predict each household's optimal choice from the entire menu $\cD$, and compute each household's utility gain (in certainty equivalent terms).
To carry out this exercise, we need to take a stand on how does the household value alternatives.
For the EU type, we use their choice utility (also called decision utility), i.e., the CARA utility function (with $\Epar$ distributed according to our estimate of the distribution $F$).
For the DT types, we report results both for their choice utility, i.e., using the Prelec distortion function in Eq.~\eqref{eq:Yaari} (with $\Ypar$ distributed according to our estimate of the distribution $G$); and for the case where the probability distortion function is completely removed, so that $\Omega(\mu)=\mu$ and the household values alternatives based on their net present value (NPV).
This also allows one to think about the effect of eliminating the mismatch between subjective and objective beliefs about loss probabilities, if this is what the probability distortion function captures.

In the second exercise,  we propose a restructuring of the auto insurance market where collision and comprehensive coverage are offered as a single auto insurance product with
\begin{align*}
\cD^{auto}&=\{100,200,250,500,1000\}\\
\mu^{auto}&=\mu^{\cl}+\mu^{\cm}\\
\covx^{\ell\,auto}&= \covx^{\cl\ell}+\covx^{\cm\ell}
\end{align*}
where $\mu^{auto}$ is the probability of experiencing a claim in either collision or comprehensive (we disregard the probability that a claim occurs in both contexts within the policy period as this probability is extremely low in our data) and $\covx^{\ell\,auto}$ is the premium charged for an auto coverage that offers the same deductible in collision and comprehensive when firms operate under perfect competition or if they use a constant markup rule.

Again, we take the preferences estimated using our limited consideration model, predict each household's optimal choice, and compute the household's utility gain/loss (in certainty equivalent terms). 
However, to carry out the exercise not only do we need to take a stand on how does the household value alternatives, but, importantly, also on how does the household draw its consideration set after the intervention.
For the former, we proceed as in our first welfare exercise, and report results where the EU types value alternatives based on their choice utility, and DT types based on both their choice utility and on the alternatives' NPV.
For the latter, we report our results under several scenarios, detailed below.
This exercise may help inform the debate on the need to ``simplify insurance choice," and clarify the role of limited consideration in mediating nudging effects.

Before presenting the results of these two exercises, we explain why EU and DT households may respond differently to an intervention that combines collision and comprehensive into a single coverage. 
A defining feature of the DT model is that it is non-linear in probabilities. 
Hence, offering insurance as a bundle or as a single product may have a first order impact on DT households' choices and welfare. 
To see why, suppose the probability distortion function is strictly sub-additive (as is the case in our estimated model). 
Then, under the maintained assumption of narrow bracketing (Assumption \ref{ass:narrow}), the agent's willingness to pay to avoid a $\$500$ loss which occurs with a 10 percent chance, is strictly lower than twice their willingness to pay to avoid the same loss with 5 percent chance. 
Put differently, a single insurance product against two (mutually exclusive) identical losses, instead of a bundle of two products, reduces the degree of over-weighting of loss probabilities. 
At the same time, combining insurance products into one line of insurance limits choice, and may eliminate the first best alternative. 
Ceteris paribus, for a fully rational agent making choices according to the EU model, this can only be welfare reducing. 
Interestingly, there are examples of insurance products that are indeed sold both as a single coverage and as a bundle, such as single limit liability coverage versus bodily injury and property damage in auto insurance. 
\begin{table}\caption{Welfare implications of limited consideration and of combining collision and comprehensive into a single coverage (95\% confidence intervals in square brackets)}\label{Table_welfare}
	\centering  \resizebox{\textwidth}{!}{
\begin{tabular}{lcccccccc}
	   \hline\hline
& \multicolumn{2}{c}{\textbf{Choice}}& \multicolumn{2}{c}{\textbf{CARA}} & \multicolumn{4}{c}{\textbf{As a \% of Average price}}\\ 
& \multicolumn{2}{c}{\textbf{Utility}}& \multicolumn{2}{c}{\textbf{NPV}} & \multicolumn{4}{c}{\textbf{for (1000,1000) (\$238)}}\\ 

\hline 
\textbf{All at Full Consideration} & \multicolumn{2}{c}{30.1} & \multicolumn{2}{c}{18.2} & \multicolumn{2}{c}{12.7} & \multicolumn{2}{c}{7.6}\\ 
& [28.3 & 31.9] & [16.0 & 20.3] & [11.9 & 13.4] & [6.7 & 8.5] \\ 
 \hline
\textbf{Bundled Auto Insurance}: &  &  &  &  &  &  &  &  \\ 
\\
Worst Case Consideration & \multicolumn{2}{c}{-3.2} &\multicolumn{2}{c}{-18.7} & \multicolumn{2}{c}{-1.3}&\multicolumn{2}{c}{-7.9} \\ 
& [-8.0  & 1.5] & [-22.2 & -15.2] &[ -3.4 & 0.6] & [-9.3 & -6.4] \\ 
Middle Case Consideration & \multicolumn{2}{c}{48.1}&\multicolumn{2}{c}{26.5} & \multicolumn{2}{c}{20.2} &\multicolumn{2}{c}{11.1} \\ 
& [45.9 & 50.3] & [25.4 & 27.7] & [19.3 & 21.1] & [10.7 & 11.7] \\ 
All at Full Consideration & \multicolumn{2}{c}{52.3}& \multicolumn{2}{c}{29.6} &\multicolumn{2}{c}{22} &\multicolumn{2}{c}{12.5} \\ 
& [50.0 & 54.7] & [28.2 & 31.1] & [21 & 23] & [11.8 & 13.1] 
 \\       \hline \hline
\end{tabular}}
\end{table}
In summary, our first welfare exercise addresses the question: what is the (average) welfare cost associated with limited consideration? 
Our second welfare exercise addresses the question: what are the welfare implications of combining collision and comprehensive into a single product, and how does the presence of limited consideration alter these implications? 

The top panel of Table \ref{Table_welfare} reports our estimates of the welfare losses due to limited consideration. 
Using the choice utility for each preference type, the welfare losses are about \$30, or 12.7\% of the average price of the cheapest bundle. 
The effect is smaller (\$18 or 7.6\%) if for DT types we use the alternatives' NPV as their value (i.e., we shut down the probability distortion).
This is expected, since all utilities and utility differences decrease.

The bottom panel of Table \ref{Table_welfare} reports estimated welfare changes associated with combining collision and comprehensive insurance into a single product.
We carry out the exercise for three different ways in which consideration sets may be drawn after the market intervention.
In the worst case scenario, in the sense that consideration is lowest, the probability that deductible $\dor$ is considered equals the estimated consideration probability for bundle $(\dor,\dor),\dor\in\cD^{auto}$. 
In this case, the impact of the intervention is negative, although the magnitude of the effect depends substantially on how the welfare of DT types is evaluated.
This is because under choice utility, following the intervention, DT types overweight the overall loss probability to a lesser degree than they did with separate coverages, and this effect attenuates substantially the welfare reduction from not being able to choose from a larger menu.
On the other hand, when welfare of DT types is evaluated according to NPV, although the overweighting of loss probabilities affects choice, it does not enter the welfare calculations. 

Under full consideration, the best case scenario, the welfare gains for both evaluation approaches are positive and large. 
Relative to the worst case scenario, this is, of course, expected. 
What is more interesting is that the welfare gains are higher than those obtained in the counterfactual of full consideration that maintains the status-quo separation between collision and comprehensive insurance.
This is because under full consideration, the EU types are worse off when the collision and comprehensive are combined into a single product (for them, the choice set is being reduced without any associated benefit); however, the DT types, despite facing a smaller choice set, benefit from such a reduction because in making choices they overweight losses by a smaller degree. 
The latter effect dominates, more so when welfare is computed based on choice utility rather than on NPV. 

For completeness we also report welfare changes for a case that we label ``middle consideration,'' in which each deductible in the combined single coverage is considered with a probability equal to the sum of the probability that it is considered either as collision or comprehensive deductible (or with probability one if the sum exceeds one).
The results are reported in the middle row of the bottom panel of Table \ref{Table_welfare}.
Even with this intermediate consideration level, the welfare gains are substantial.

Based on these welfare exercises, we argue that the interplay between features of the decision making process at the utility evaluation level and of the consideration mechanism cannot be ignored when analyzing possible market interventions. 
In the second welfare exercise carried out above, reducing the feasible set may lead to unambiguous welfare gains, provided consideration increases. 
However, if consideration does not increase, the same intervention can lead to welfare losses that exceed the gains stemming from nudging the non-expected utility maximizers in the population.    

\section{Discussion}\label{sec:discussion}
This paper provides semi-nonparametric point identification results for a model of discrete choice under risk that allows for unobserved heterogeneity in preference types,  unobserved heterogeneity within each type, and unobserved heterogeneity in consideration sets, while confronting the fact that the covariates $\covx$ characterizing products do not exhibit independent variation across alternatives within a context, but only across contexts.
We apply our method to study demand for deductible insurance in two lines of property insurance, and to analyze the welfare implications of an hypothetical market intervention where the two lines of insurance are combined into a single one.
Our findings provide evidence of the importance of allowing for the rich amount of unobserved heterogeneity that our model features.

The choice environment that we study in this paper is similar to that studied in \cite{BaMoTh21}.
They offer a comprehensive analysis of the implications of the Spence-Mirlees single crossing property for semi-nonparametric identification of a model of discrete choice under risk that features a single preference type and unobserved heterogeneity in consideration sets.
They also illustrate the tradeoff between the common exclusion restrictions and the restrictions on consideration set formation required for semi-nonparametric point identification.
Their work is the closest to ours.
However, in our model consideration sets are formed at the bundle level (i.e., across contexts), and hence the single crossing property that both \cite{BaMoTh21} and we assume to hold within a context, may not necessarily hold across tuples of alternatives.
This is because bundles may not be monotonically ranked (with respect to preference parameters) against each other.
Hence, the results in \citet{BaMoTh21} do not apply and in this paper we develop a new approach to obtain point identification of the distribution of preferences, of the shares of preferences types, and of features of the distribution of consideration sets given type.\footnote{As we allow for multiple preference types, our analysis extends that of \citet{BaMoTh21} even in the simplified framework where consideration is independent across contexts.}
In \cite{BM23} we show that in a richer data environment where the researcher observes a characteristic for each alternative that displays independent variation both across agents and across alternatives, and that affects utility but not consideration,  semi-nonparametric point identification holds for a flexible pure random coefficients model with unrestricted dependence between the random coefficients and the consideration set formation mechanism.

The challenges posed to identification of discrete choice models by unobserved heterogeneity in consideration sets have long been recognized \citep[e.g.,][]{Manski1977}.\footnote{Many important papers in the theory literature---including papers on revealed preference analysis under limited attention, limited consideration, rational inattention, and other forms of bounded rationality that manifest in unobserved heterogeneity in consideration sets---also grapple with the identification problem \citep[e.g.,][]{masatioglu_12,man:mar14,Caplin2015,Lleras2017,Cattaneo2019}. However, these papers generally assume rich datasets---e.g., observed choices from every possible subset of the feasible set---that often are not available in applied work, especially outside of the laboratory. }
It is not uncommon for the problem to be ignored, as a textbook assumption is that agents pick an alternative to maximize their utility over the entire feasible set.
When heterogeneity in consideration sets is allowed for, point identification of the model often relies on the availability of auxiliary information about the composition or distribution of agents' consideration sets, or on two-way exclusion restrictions, whereby certain variables impact consideration but not preferences and vice versa.
A third approach relies primarily on restrictions to the consideration set formation process.\footnote{Examples for the first approach include \citet{DelosSantos2012, Conlon2013,HonkaRAND2017,HonkaMS2017}; for the second, \citet{Goeree2008,vanNierop2010,Gaynor2016,Heiss2016}.
Recent examples for the third approach include \citet{Abaluck2019,Crawford2019,Lu2018}.}

When such assumptions may not be credible and one does not have access to auxiliary data or valid exclusion restrictions, \cite{BCMT21} provide a method to obtain informative sharp identification regions for the parameters of discrete choice models, even when preferences and consideration sets may depend on each other, under the assumption that agents' consideration sets include at least two alternatives.
\citet{Cattaneo2019,cat:che:ma:mas21} provide revealed preference theory, testable implications, and partial identification results for preference orderings and attention frequency, in very general models of limited consideration but without heterogeneity in preferences, under the assumption that one observes agents repeated choices (in a single context) while facing varying choice sets.

\appendix
\section{Appendix: Proof of Theorem \ref{T1alt} and Corollary \ref{cor:T1}}
\begin{proof} [Proof of Theorem \ref{T1alt}]\label{T1p_alt} 

Fix $\Epar\in[\Epar^*,\Epar^{**}]$ and the corresponding $\Covx^1(\Epar)$.
By Assumption \ref{ass:var_x}, $\Covx^1(\Epar)$ is non-empty and there is an $\epsilon$-ball around it of positive density.
By Definition \ref{def:indifferenceX}, for any $(\covx^\cl,\covx^\cm)\in \Covx^1(\Epar)$, $\EEcut^{1,1}_{2,1}(\covx^\cl)=\EEcut^{1,1}_{1,2}(\covx^\cm)=\Epar$.
Along with Assumption \ref{ass:distinct_contexts}, this implies that there are vectors $\covx'=(\covx^{\cl \prime},\covx^{\cm \prime}) \in \B_\epsilon(\Covx^1(\Epar))$ and $\covx''=(\covx^{\cl \prime \prime},\covx^{\cm \prime \prime})\in \B_\epsilon(\Covx^1(\Epar))$ such that $\Epar=\EEcut^{1,1}_{2,1}(\covx^{\cl \prime})<\EEcut^{1,1}_{1,2}(\covx^{\cm \prime})$ and $\EEcut^{1,1}_{2,1}(\covx^{\cl \prime \prime})>\EEcut^{1,1}_{1,2}(\covx^{\cm \prime \prime})$. 
We claim that
\begin{align}\label{T1p_eq_1}
	\lim\limits_{\covx',\covx''\rightarrow \covx}\left(\frac{\partial \Pr(\cI^*=\cI_{1,1}|\covx')}{\partial \covx^{\cl}}-\frac{\partial \Pr(\cI^*=\cI_{1,1}|\covx'')}{\partial \covx^{\cl}} \right)=\alpha f(\Epar) \cdot \cSE(\covx,\cOE),
\end{align}
where $\cSE(\covx,\cOE)$ is a function of $\covx$ and of the consideration probabilities given by:
\begin{align}
	\cSE(\covx,\cOE) & = \cOE(\{\cI_{1,1},\cI_{2,1}\};\emptyset)\frac{\partial \EEcut^{1,1}_{2,1}(\covx)}{\partial \covx^{\cl}}+\cOE(\{\cI_{1,1},\cI_{2,2}\}; \cI_{2,1})\frac{\partial \EEcut^{1,1}_{2,2}(\covx)}{\partial \covx^{\cl}}\notag\\
	& - \left(\cOE(\{\cI_{1,1},\cI_{2,2}\};\cI_{1,2})\frac{\partial \EEcut^{1,1}_{2,2}(\covx)}{\partial \covx^{\cl}}+\cOE(\{\cI_{1,1},\cI_{2,1}\}; \{\cI_{2,2},\cI_{1,2}\}) \frac{\partial \EEcut^{1,1}_{2,1}(\covx)}{\partial \covx^{\cl}}\right)\label{eq:S_Epar}
\end{align}
Under Assumption \ref{ass:MIC}-\ref{MIC1},  Eq.~\eqref{eq:S_Epar} simplifies to\footnote{These derivations are based on repeated use of facts such as
\begin{align*}
\cOE(\{\cI_{1,1},\cI_{2,2}\}; \emptyset)&=\cOE(\{\cI_{1,1},\cI_{2,2},\cI_{2,1}\}; \emptyset)+\cOE(\{\cI_{1,1},\cI_{2,2}\}; \{\cl_{2,1}\})
=\cOE(\{\cI_{1,1},\cI_{2,2},\cI_{1,2}\}; \emptyset)+\cOE(\{\cI_{1,1},\cI_{2,2}\}; \{\cl_{1,2}\})
\end{align*}
}
\begin{align}
	\cSE(\covx,\cOE) 
	 & =\Big(\cOE(\{\cI_{1,1},\cI_{2,1}\};\emptyset)-\cOE(\{\cI_{1,1},\cI_{2,1}\}; \{\cI_{2,2},\cI_{1,2}\})\Big)\frac{\partial \EEcut^{1,1}_{2,1}(\covx)}{\partial \covx^{\cl}}\neq 0\label{eq:S_MIC1}
\end{align}
Under Assumption \ref{ass:MIC}-\ref{MIC2},  Eq.~\eqref{eq:S_Epar} simplifies to
\begin{align}
	\cSE(\covx,\cOE) & =\Big(\cOE(\{\cI_{1,1},\cI_{2,2}\}; \cI_{2,1})-\cOE(\{\cI_{1,1},\cI_{2,2}\}; \cI_{1,2})\Big)\frac{\partial \EEcut^{1,1}_{2,2}(\covx)}{\partial \covx^{\cl}}\neq 0\label{eq:S_MIC2}
\end{align}

To derive the expression for $\cSE(\covx,\cOE)$ in Eq.~\eqref{eq:S_Epar}, we return to Eq.~\eqref{Derivate_ij}, which states
\begin{align*}
		\frac{\partial\Pr(\cI^*=\cI_{1,1}|\covx)}{\partial \covx^{\cl}} 
	&= \alpha \sum\limits_{(k,r)\neq(1,1)} \cOE(\{\cI_{1,1},\cI_{k,r}\};\cB(\cI_{1,1},\covx;\EEcut^{1,1}_{k,r}))f(\EEcut^{1,1}_{k,r})\frac{\partial \EEcut^{1,1}_{k,r}}{\partial \covx^{\cl}}\notag\\
	&+(1-\alpha) \sum\limits_{(k,r)\neq(1,1)} \cOY(\{\cI_{1,1},\cI_{k,r}\};\cB(\cI_{1,1},\covx;\YYcut^{1,1}_{k,r}))g(\YYcut^{1,1}_{k,r})\frac{\partial \YYcut^{1,1}_{k,r}}{\partial \covx^{\cl}}
\end{align*}
Under Assumptions \ref{ass:distinct_contexts}-\ref{distinct_sufficient} and \ref{ass:var_x}, when $\covx'$ and $\covx''$ are sufficiently close to $\covx$, the relative order of the cutoffs for type $t_i=0$ preferences, $\YYcut^{1,1}_{k,r}$, does not change. 
For type $t_i=1$ preferences, it changes only for the cutoffs involving bundles $\{\cI_{2,1},\cI_{1,2},\cI_{2,2}\}$. 
Hence,
\begin{align}
	\frac{\partial\Pr(\cI^*=\cI_{1,1}|\covx')}{\partial \covx^{\cl}} 
	&= \alpha \cOE(\{\cI_{1,1},\cI_{2,1}\};\emptyset)f(\EEcut^{1,1}_{2,1})\frac{\partial \EEcut^{1,1}_{2,1}(\covx')}{\partial \covx^{\cl}}\notag\\
	&+\alpha \cOE(\{\cI_{1,1},\cI_{2,2}\};\cI_{2,1})f(\EEcut^{1,1}_{2,2})\frac{\partial \EEcut^{1,1}_{2,2}(\covx')}{\partial \covx^{\cl}}+R(\covx')\label{eq:derivative_x}\\
	\notag\\
	\frac{\partial\Pr(\cI^*=\cI_{1,1}|\covx'')}{\partial \covx^{\cl}} 
	&= \alpha \cOE(\{\cI_{1,1},\cI_{2,2}\};\cI_{1,2})f(\EEcut^{1,1}_{2,2})\frac{\partial \EEcut^{1,1}_{2,2}(\covx'')}{\partial \covx^{\cl}}\notag\\
	&+\alpha \cOE(\{\cI_{1,1},\cI_{2,1}\};\{\cI_{1,2},\cI_{2,2}\})f(\EEcut^{1,1}_{2,1})\frac{\partial \EEcut^{1,1}_{2,1}(\covx'')}{\partial \covx^{\cl}}+R(\covx'')  \label{eq:derivative_xx}
\end{align}
where $R(\cdot)$ is a collection of terms that are continuous functions of their argument around $\covx$. 
Consequently, in the limit where both $\covx',\covx''$ tend to $\covx$, $R(\covx')$ and $R(\covx'')$ are identical to each other, and Eq.~\eqref{eq:S_Epar} follows by subtracting Eq.~\eqref{eq:derivative_xx} from Eq.~\eqref{eq:derivative_x}.

Next, observe that $\cSE(\covx,\cOE)$ equals a non-zero constant multiplied with $\frac{\partial \EEcut^{1,1}_{2,1}(\covx)}{\partial \covx^{\cl}}$ (or $\frac{\partial \EEcut^{1,1}_{2,2}(\covx)}{\partial \covx^{\cl}}$).
The latter term is a known function of the data and is different from zero.  
Consequently, the density function of the random coefficient for type $t_i=1$ agents evaluated at $\Epar$, $f(\Epar)$, is identified up to a non-zero constant ($\alpha$ multiplied with a non-zero linear combination of consideration probabilities that does not depend on $\Epar$). 
If $[\Epar^*,\Epar^{**}]=[0,\bar\Epar]$, then using that $f(\Epar)$ integrates to one over its support identifies $\alpha \cdot \cSE(\covx,\cOE)$, and consequently the entire function $f(\cdot)$. 
The same argument applies to establish identification of $g(\cdot)$.
\end{proof}

\begin{proof}[Proof of Corollary \ref{cor:T1}]\label{cor:T1proof}  
Once $\alpha f(\Epar) \cdot \cSE(\covx,\cOE)$ and $(1-\alpha) g(\Ypar) \cdot \cSY(\covx,\cOY)$ are identified, so are $f(\Epar)$ and $g(\Ypar)$ provided there is large support. 
Under Assumption \ref{ass:MIC}, $\cSE(\covx,\cOE)$ and $\cSY(\covx,\cOY)$ can be decomposed into a product of two terms, one known and another entirely dependent on consideration, see Eqs.~\eqref{eq:S_MIC1}-\eqref{eq:S_MIC2}. 
Moreover, these terms will be identical, as long as $\cOE(\cdot;\cdot)=\cOY(\cdot;\cdot)$ for all relevant combinations of $\{\cI_{1,1},\cI_{1,2},\cI_{2,1},\cI_{2,2}\}$ in part (i) (respectively, part (ii)) of the assumptions stated in Corollary \ref{cor:T1}. 
Hence, the ratio of $\alpha f(\Epar) \cdot \cSE(\covx,\cOE)$ and $(1-\alpha) g(\Ypar) \cdot \cSY(\covx,\cOY)$ identifies $\alpha$.
\end{proof}

	\setstretch{1}

\bibliography{bcmt,welfare,BMT_BCMT}

\end{document}